\documentclass[conference, 10pt]{IEEEtran}
\ifCLASSINFOpdf
\else
\fi

\usepackage{amsmath}
\usepackage{amsthm}
\usepackage{amssymb}
\usepackage{latexsym}
\usepackage{proof}
\usepackage{microtype}

\usepackage{bussproofs}
\usepackage{amssymb}
\usepackage{amsmath}
\usepackage{mathrsfs}
\usepackage{upgreek}
 \usepackage[bbgreekl]{mathbbol}
 \usepackage{graphicx}

\usepackage{color}

\newdimen\CdotAxis
\newcommand*{\CdotAux}[3]{%
  {%
    \settoheight\CdotAxis{$#2\vcenter{}$}%
    \sbox0{%
      \raisebox\CdotAxis{%
        \scalebox{#1}{%
          \raisebox{-1.1pt}{%
            $\mathsurround=0pt #2#3$%
          }%
        }%
      }%
    }%
    \dp0=0pt %
    \sbox2{$#2\bullet$}%
    \ifdim\ht2<\ht0 %
      \ht0=\ht2 %
    \fi
    \sbox2{$\mathsurround=0pt #2#3$}%
    \hbox to \wd2{\hss\usebox{0}\hss}%
  }%
}

\newcommand{\comment}[1]{}

\newcommand{\D}                       { {\mathsf{com}} }

\newcommand{\LC}     {{\mathsf{G}}}

\newcommand{\Language}                 {\mathcal{L}}

\newcommand{\proj}                     { {\mathsf{\uppi}} }
\newcommand{\inj}                   {{{\upiota}}}
\newcommand{\emp}[1]    {{\mathsf{P}_{#1}}}

\newcommand{\Hyp}[3]                   {{\mathtt{H}_{#1}^{\forall {#3} {#2}}}}

\newcommand{\wit}                 {{\mathsf{wit}_{\emp{}}}}
\newcommand{\Wit}[3]               {{\mathtt{W}_{#1}^{\exists {#3} {#2} }}}

\newcommand{\pair}[2]{\langle #1,#2\rangle}

\newcommand{\nf}{\mathsf{NF}}

\newcommand{\Ecrom}[3]{#2 \parallel_{#1} #3}

\newcommand{\abort}    {{\mathcal{A}}}

\newcommand{\efq}[2]   {{\mathsf{efq}_{#1}(#2)}}

\newcommand{\sml}[1] {{\scriptscriptstyle #1}}
\newcommand{\sq}[1]  {{\boldsymbol{#1}}}
\newcommand{\nor}[1] {{\mathcal{#1}}}

\newcommand{\DER}{\vdash}
\newcommand{\THEN}{\Rightarrow}
\newcommand{\IMPL}{\rightarrow}
\newcommand{\ET}{\wedge}
\newcommand{\OR}{\vee}
\newcommand{\FAL}{\bot}

\newcommand{\SET}{\bigwedge}

\newcommand{\COM}{(\mathsf{com})}
\newcommand{\ECONG}{\mathrm{E}}
\newcommand{\ICONG}{\mathrm{I}}

\newcommand{\COMB}{\mathsf{com}}
\newcommand{\COML}{\mathsf{com}_{l}}
\newcommand{\COMR}{\mathsf{com}_{r}}

\newcommand{\lamg}{\bblambda_{\mathrm{G}}}
\newcommand{\lamgtitle}{\bblambda_{\mathrm{G}}}

\newcommand{\NJ}{\mathbf{NJ}}
\newcommand{\NGod}{\mathbf{NG}}

\theoremstyle{plain}
\newtheorem{theorem}{Theorem}[section]
\newtheorem{lemma}[theorem]{Lemma}
\newtheorem*{theorem*}{Theorem}
\newtheorem{proposition}[theorem]{Proposition}

\theoremstyle{definition}
\newtheorem{definition}{Definition}[section]
\newtheorem{example}{Example}[section]
\theoremstyle{remark}
\newtheorem{remark}{Remark}[section]

\usepackage{xcolor}

\usepackage{tikz}
\usetikzlibrary{arrows}
\usetikzlibrary{decorations.pathmorphing}

\def\NJ{{\bf NJ}}

\newcommand{\seq}{\Rightarrow}

\def\hh{\ |\ } 
\def\Deduce#1{\hbox{$\hphantom{#1}$\kern\inferLabelSkip\DeduceSym\kern\inferLab
elSkip$#1$}}

\newcommand{\logic}[1]{\ensuremath{\mathbf #1}} 
\def\LC{\logic{G}}

\begin{document}
\IEEEoverridecommandlockouts

\title{G\"odel Logic:
from Natural Deduction to Parallel Computation}



\author{\IEEEauthorblockN{Federico Aschieri
}
\IEEEauthorblockA{Institute of Discrete Mathematics and Geometry\\TU Wien, Austria\\
}
\and
\IEEEauthorblockN{Agata Ciabattoni
}
\IEEEauthorblockA{Theory and Logic Group\\TU Wien, Austria\\
}
\and
\IEEEauthorblockN{Francesco A. Genco
}
\IEEEauthorblockA{Theory and Logic Group\\TU Wien, Austria\\
}
\thanks{Supported by FWF: grants M 1930--N35, Y544-N2,
and W1255-N23.}
}



%


\maketitle

\begin{abstract}
Propositional G\"{o}del logic $\LC$ extends intuitionistic logic with the non-constructive principle of linearity \mbox{$(A \IMPL B ) \OR (B \IMPL A)$}. 
We introduce a Curry--Howard correspondence for 
$\LC$ and show that a simple natural deduction 
calculus can be used as a typing system. 
The resulting functional language extends the simply typed $\lambda$-calculus via a synchronous communication mechanism between parallel processes, which increases its expressive power.
The normalization proof employs original
termination arguments and proof transformations 
implementing forms of code mobility. 
Our results provide a computational interpretation of $\LC$, thus 
proving  
A.\ Avron's 1991 thesis.
\end{abstract}




\IEEEpeerreviewmaketitle

\section{Introduction}
 Logical proofs  are static. Computations are dynamic.
It is a striking discovery that the two coincide: formulas correspond
to types in a programming language, logical proofs to programs of the corresponding types and removing detours from proofs to
 evaluation of programs.
This correspondence, known as Curry--Howard isomorphism, 
was first discovered for constructive proofs, and in particular for intuitionistic natural deduction and 
typed $\lambda$-calculus \cite{Howard} and
later extended to classical proofs, despite their use of non-constructive principles,
such as the excluded middle \cite{deGrooteex, AschieriZH} or reductio ad absurdum \cite{Griffin, Parigot}. 

Nowadays various different logics (linear~\cite{CP2010}, modal~\cite{MCHP} ...) have been related to many different notions of  computation; the list is long, and we refer the reader to \cite{Wadler}.

\subsection*{G\"odel logic, Avron's conjecture and previous attempts}
Twenty-five years have gone by since Avron conjectured
in \cite{Avron91} that G\"odel logic $\LC$
\cite{goedel} -- one of the most useful and interesting logics
intermediate between intuitionistic and classical logic --
might provide a basis for parallel $\lambda$-calculi. 
Despite the  interest of the conjecture and despite various attempts, 
no Curry--Howard correspondence has so far been
provided for $\LC$. The main obstacle has been the lack of an
adequate natural deduction calculus. 
Well designed natural deduction inferences can indeed be  
naturally 
interpreted as program instructions, in particular as typed
$\lambda$-terms. Normalization~\cite{Prawitz}, which corresponds to
the execution of the resulting programs, can then be used to obtain
proofs only containing formulas that are
subformulas 
of some of the hypotheses or of the conclusion. However the
problem  of finding a natural deduction for $\LC$ with this
property, called analyticity, looked hopeless for decades.

All approaches explored so far to
provide a precise formalization of
$\LC$ as a logic for parallelism, either sacrificed
analyticity~\cite{  Aschieri2016}
 or tried to devise forms
of natural deduction whose structures mirror 
{\em hypersequents} -- which are sequents 
operating in parallel \cite{Avron96}. 
Hypersequents were indeed successfully
used in \cite{Avron91} to define an analytic calculus for $\LC$ and were intuitively 
connected to parallel computations: the key rule introduced by
Avron to capture the linearity axiom -- called {\em communication} -- enables
sequents to exchange their information and hence to ``communicate''.
The first analytic natural
  deduction calculus proposed for $\LC$ \cite{HyperAgata} uses indeed parallel intuitionistic derivations  
  joined together by the hypersequent separator. 
Normalization is obtained
 there only by translation into Avron's calculus: no
reduction rules for deductions and no
corresponding $\lambda$-calculus were provided. The former task
was carried out in \cite{BP2015}, that contains
a propositional hyper natural deduction with a normalization
procedure. The definition of a corresponding $\lambda$-calculus and
Curry--Howard correspondence are left as an open problem, which might
have a complex solution due to the elaborated structure of hyper
deductions.  Another attempt along the ``hyper line'' has been made in
\cite{Hirai}. However, not only the proposed proof system is not shown
to be analytic, but the associated $\lambda$-calculus is not a
Curry--Howard isomorphism: the computation rules of the 
$\lambda$-calculus are not related to proof transformations, i.e.\ {\em
Subject Reduction} does not hold.
\subsection*{$\lamg$: Our Curry--Howard Interpretation of G\"odel Logic}
We introduce a natural deduction and a Curry--Howard correspondence for propositional $\LC$. 
We add to the $\lambda$-calculus an operator that, from the programming viewpoint, represents
parallel computations and communications between them; from the logical viewpoint, the linearity axiom; and from the proof theory viewpoint, the
hypersequent separator among sequents.
We call the resulting calculus $\lamg$: parallel $\lambda$-calculus
for $\LC$. $\lamg$ relates to the natural deduction $\NGod$ for $\LC$
as typed $\lambda$-calculus relates to the natural deduction $\NJ$ for
intuitionistic logic $\mathbf{IL}$:
\tikzstyle{crc}=[circle, minimum size=7mm, inner sep=0pt, draw]
\begin{center}
\begin{tikzpicture}[node distance=1.5cm,auto,>=latex', scale=0.3]

\node [crc, align=center] (1) at (-10,3.5) {$\mathbf{IL}$};

\node[crc, align=center] (2) at (0,3.5) {$\NJ$};

\node [crc,align=center] (3) at (10,3.5) {$\lambda$};

\node [crc, align=center] (4) at (-10,0) {$\LC$};

\node[crc, align=center] (5) at (0,0) {$\NGod$};

\node [crc,align=center] (6) at (10,0) {$\lamg$};

\path[<->] (1) edge [thick, align=center] node {}(2);

\path[<->] (2) edge [double, thick, align=center] node
{}(3);

\path[<->] (4) edge [thick, align=center] node {\emph{Soundness
and}\\\emph{Completeness}}(5);

    \path[<->] (5) edge [double, thick, align=center] node {\emph{Curry--Howard}\\ \emph{correspondence}}(6);
  \end{tikzpicture}
\end{center}
We prove: the perfect match between computation steps and proof reductions in the Subject Reduction Theorem; the Normalization Theorem, by providing a  terminating reduction strategy for $\lamg$; the Subformula Property, as corollary.
The expressive power of $\lamg$ is illustrated through examples of programs and
connections with the $\pi$-calculus~\cite{Milner, sangiorgiwalker2003}.

The natural deduction calculus $\NGod$ that we use as type system for
$\lamg$ is particularly simple: it extends $\NJ$ with the $\COM$ rule
(its typed version is displayed below), which was first considered in
\cite{L1982} to define a natural deduction calculus for $\LC$, but
with no normalization procedure.  The calculus $\NGod$ follows the
basic principle of natural deduction that \emph{new axioms require new
computational reductions}; this contrasts with the basic principle of
sequent calculus employed in the ``hyper approach'', that \emph{new
axioms require new deduction structures}.  Hence we keep the calculus
simple and deal with the complexity of the hypersequent structure at
the operational side.  Consequently, the programs corresponding to
$\NGod$ proofs maintain the syntactical simplicity of
$\lambda$-calculus.  The normalization procedure for $\NGod$ extends
Prawitz's method with ideas inspired by hypersequent cut-elimination,
by normalization in classical logic \cite{AschieriZH} and by the
embedding in \cite{CG2016} between hypersequents and systems of rules
\cite{Negri:2014}; the latter shows that $\COM$ reformulates Avron's
communication rule.

The inference rules of $\NGod$ are decorated with $\lamg$-terms, so that we can directly read proofs as typed programs. The decoration of the 
$\NJ$ inferences is standard and the typed version of $\COM$ is 
\[
\infer[\D]{u \parallel_{a} v: C}{\infer*{u:C}{[a^{\scriptscriptstyle A\rightarrow B}: A\rightarrow B]} && \infer*{v:C}{[a^{\scriptscriptstyle B\rightarrow A}: B\rightarrow A]}}\]
 Inspired by~\cite{Aschieri2016}, we use the variable $a$ to represent a {\em private} communication
channel between the processes $u$ and $v$.
The computational reductions associated to 
$\parallel_{a}$ -- \emph{cross
reductions} -- 
enjoy a natural interpretation in terms of 
higher-order process passing, a feature which is not directly rendered through
communication by reference~\cite{perez2015} and is also present 
in higher-order $\pi$-calculus~\cite{sangiorgiwalker2003}.
Nonetheless cross reductions handle more subtle migration issues. 
In particular, a cross reduction can be activated whenever a
communication channel $a$ is ready to transfer information between two
parallel processes:
\[ \mathcal{C}[a\, u]\parallel_{a} \mathcal{D}[a\, v]
\]
Here $\mathcal{C}$ is a process containing a fragment of code $u$, and
$\mathcal{D}$ is a process containing a fragment of code
$v$. Moreover, $\mathcal{C}$ has to send $u$ through the channel $a$
to $\mathcal{D}$, which in turn needs to send $v$ through $a$ to
$\mathcal{C}$.
In general we cannot simply
send the programs $u$ and $v$: some
resources in the computational environment that are used by  $u$ and $v$
may become inaccessible from the new locations~\cite{codemobil}. 
Cross reductions solve the problem by exchanging the location of $u$ and $v$ and creating a new communication channel for
their resources. Technically, the channel takes care of \emph{closures} -- the contexts containing the definitions of the
variables used in a function's body \cite{landin}. Several 
programming languages such as JavaScript, Ruby or Swift provide 
mechanisms to support and handle closures.
In our case, they are the basis of a
\emph{process migration mechanism} handling the bindings between code fragments and their computational environments.
Cross reductions also improve the efficiency of programs by facilitating
partial evaluation of open processes (see Example~\ref{ex:code_mobility}).

\section{Preliminaries on G\"odel logic}
Also known as G\"odel--Dummett logic~\cite{dummett}, 
G\"odel logic $\LC$
naturally turns up in a number of different contexts; among them, 
due to the natural interpretation of its connectives as functions over the real interval $[0, 1]$,
$\LC$ is one of the best known `fuzzy logics', e.g.~\cite{MGO}.

Although  propositional $\LC$ is obtained by adding
the linearity axiom $(lin) $ $(A \IMPL B) \OR (B \IMPL A)$ 
to any proof calculus for intuitionistic logic,
analytic calculi for $\LC$ have only been defined in formalisms {\em extending} the sequent calculus.
Among them, arguably, the hypersequent calculus in \cite{Avron91} is the most successful 
one, see, e.g.,~\cite{MGO}. 
In general a hypersequent calculus is defined by incorporating
a sequent calculus (Gentzen's {\em LJ}, in case of $\LC$) as a
sub-calculus and allowing sequents to live in the context of finite
multisets of sequents. 
\begin{definition}
A \textbf{hypersequent} is a multiset of sequents, written as
$  \Gamma_1 \seq \Pi_1 \hh \dots \hh \Gamma_n \seq \Pi_n$
where, for all $i = 1, \dots n,$ $\Gamma_i \seq \Pi_i$ is an ordinary
sequent.  
\end{definition}
The symbol ``$|$'' is a meta-level disjunction; this is reflected by the 
presence in the calculus of the external structural rules of 
weakening and contraction, operating on whole sequents, rather than on formulas.  
The hypersequent design opens the possibility of defining new rules that allow the
``exchange of information'' between different sequents. It is this type
of rules which increases the expressive power of hypersequent calculi
compared to sequent calculi.  The additional rule employed in
Avron's calculus for $\LC$ \cite{Avron91} is the so called {\em
communication rule}, below presented in a slightly reformulated version
(as usual $G$ stands for a possibly empty hypersequent):

{\small
\[
\infer{G\hh \Gamma_1, A \seq C \hh  
\Gamma_2, B \seq D }{G \hh \Gamma_1, B \seq C & 
G\hh \Gamma_2, A \seq D}\]
}

\vspace{-0.5cm}
\section{Natural Deduction}\label{section-ND}
The very first step in the design of a Curry--Howard correspondence is
to lay a solid logical foundation. No architectural mistake is allowed
at this stage: the natural deduction must be structurally simple and
the reduction rules as elementary as possible.
We present such a natural deduction system $\NGod$ for G\"odel logic. $\NGod$
extends Gentzen's propositional natural deduction $\NJ$ (see~\cite{Prawitz}) with a rule accounting for axiom $(lin)$. 
We describe the reduction rules for transforming every $\NGod$ deduction into an analytic one and present the ideas 
behind the Normalization Theorem, which is proved in the $\lambda$-calculus framework in Section \ref{section-normalization}.

$\NGod$ is the natural deduction version of the sequent calculus with systems
of rules in~\cite{Negri:2014}; the latter embeds (into) 
Avron's hypersequent calculus for $\LC$.
Indeed~\cite{CG2016} introduces a mapping from (and into) derivations in Avron's calculus into (and from) derivations
in the {\em LJ} sequent calculus for intuitionistic logic with the addition of the system of rules
{\small
\[ \infer[(com_{end})]{\Gamma \THEN \Pi}{\infer*{\Gamma \THEN
\Pi}{\infer[(com_{1})]{A , \Gamma_{1} \THEN C}{B, \Gamma _{1}
\THEN C}} & \infer*{\Gamma \THEN \Pi}{\infer[(com_{2})]{B,
\Gamma _{2} \THEN D}{A, \Gamma _{2} \THEN D}}}
\]
}
where $(com_{1}) , (com_{2})$ can only be applied (possibly many times) 
above respectively the left and right premise of $(com_{end})$. 
The above system, that reformulates Avron's {\em communication} rule, 
immediately translates into the natural deduction rule below, whose addition to $\NJ$ leads to
a natural deduction calculus for $\LC$
\begin{small}
\[ \infer[\D]{C}{ \infer*{C} {\infer[\COML]{ A } {
\deduce{ B } {\vdots}} } & \infer*{ C } {\infer[\COMR]{ B } {
\deduce{ A } {\vdots} } } }
\]
\end{small}
Not all the branches of a derivation containing the above rule are $\NJ$ derivations.
To avoid that, and to keep the proof of the Subformula Property (Theorem~\ref{theorem-subformula}) 
as simple as possible, we use the equivalent rule below, first considered in \cite{L1982}.
\begin{definition}[$\NGod$]
\label{def:NG}
The natural deduction calculus $\NGod$ extends $\NJ$ with the $\COM$ rule:
\[\infer[\D]{C}{\infer*{C}{[A \IMPL B]} & \infer*{C}{[B \IMPL A]}}\]
\end{definition}

Let $\DER_{\NGod}$ and $\DER_{\LC}$ indicate the derivability relations in
$\NGod$ and in $\NJ + (\textit{lin})$,
respectively.
\begin{theorem}[Soundness and Completeness] \label{th:sound}
For any set $\Pi$ of formulas and formula $A$,
 $\Pi \DER_{\NGod} A$ if and only if $\Pi \DER_{\LC} A$. 
\end{theorem}
\begin{proof} ($\Rightarrow)$ Applications of $\COM$ can be
simulated by $\OR$ eliminations having as major premiss an instance of
$(\textit{lin})$.
($\Leftarrow$) 
Easily follows by the following derivation: 
  \[ \infer[\COMB^{1}]{(A \IMPL B) \OR (B \IMPL A)}{\infer
      {(A \IMPL
        B) \OR (B \IMPL A)}{{[A \IMPL
          B]^{1}}} &&& \infer
      {(A \IMPL B)
        \OR (B \IMPL A)}{{[B \IMPL
          A]^{1}}}}
\]\end{proof}
\noindent {\em Notation}.
To shorten derivations henceforth we will use
\[
\AxiomC{$A$} \RightLabel{$\COML$} \UnaryInfC{$B$} \DisplayProof \qquad
\AxiomC{$B$} \RightLabel{$\COMR$} \UnaryInfC{$A$} \DisplayProof
\quad \text{ as abbreviations for }\]
\[
\AxiomC{$[A\rightarrow B]$} \AxiomC{$A$} \BinaryInfC{$B$}
\DisplayProof \qquad \AxiomC{$[B\rightarrow A]$} \AxiomC{$B$}
\BinaryInfC{$A$} \DisplayProof
\]
respectively, and call them \textbf{communication inferences}.

As usual, we will use $\neg A$ and $\top$ as shorthand for $A\rightarrow\bot$ and $\bot \IMPL \bot$.
Moreover, we exploit the equivalence of $A \lor B$ and 
$((A\rightarrow B)\rightarrow B)\  \land\ ((B\rightarrow A)\rightarrow
A)$ in $\LC$ (see \cite{dummett}) and treat $\OR$ as a defined connective.

\subsection{Reduction Rules and Normalization}

A normal deduction in $\NGod$ should have two essential features:
every  intuitionistic Prawitz-style reduction should have been carried out
and the Subformula Property should hold.  
Due to the $\COM$ rule, the former is not always enough to guarantee the latter. 
Here we present the main ideas behind the normalization procedure for
$\NGod$ and the needed reduction rules.
The computational interpretation of the rules
will be carried out through the $\lamg$ calculus in Section \ref{section-system}.

The main steps of the normalization procedure are as follows: 
\begin{itemize}
\item We permute down all applications of $\COM$.
\end{itemize}
The resulting deduction -- we call it in {\em parallel form} --
consists of purely intuitionistic subderivations joined together 
by consecutive $\COM$ inferences occurring immediately above the root.
This transformation is a key tool in the embedding 
 between hypersequents and systems of rules \cite{CG2016}.
The needed reductions are instances of
Prawitz-style permutations for $\vee$ elimination.
Their list can be obtained by translating into
natural deduction the permutations in Fig.~\ref{fig:red}.

Once obtained a parallel form, we interleave the following two steps.     
\begin{itemize}
\item We apply the standard intuitionistic reductions (\cite{Prawitz}) 
to the parallel branches of the derivation.
\end{itemize}
This way we normalize each single intuitionistic derivation, and this can be done
in parallel.
The resulting derivation, however, need not satisfy yet the Subformula Property.
Intuitively, the problem is that communications may discharge hypotheses that have nothing to do with their conclusion.
\begin{itemize}
\item We apply specific reductions to replace the $\COM$ applications that violate the Subformula
Property. 
\end{itemize}
These reductions -- called \emph{cross reductions} -- account for the hypersequent cut-elimination.
\label{sec:cross_explained}
They allow to get rid of the new detours that appear in configurations like the one below on the left.
To remove these detours, a first idea would be to simultaneously move the
deduction $\mathcal{D}_1$ to the right
and $\mathcal{D}_2$ to the left 
thus obtaining the derivation below right:
\[
\infer[\D]{C}{\infer*{C}{\infer[\COML]{B}{\deduce{A}{\mathcal{D}_{1}}}}&\infer*{C}{\infer[\COMR
    ]{A}{\deduce{B}{\mathcal{D}_{2}}}}}
\quad \quad \quad
\quad \quad \quad
\infer{C}{\infer*{C}{\deduce{B}{\mathcal{D}_{2}}}&\infer*{C}{\deduce{A}{\mathcal{D}_{1}}}}\]
In fact, in the context of Krivine's realizability, Danos and Krivine [9] studied the linearity axiom as a theorem of classical logic and discovered that its realizers implement a \emph{restricted} version of this transformation. Their transformation
does not lead however to the subformula property for NG.
The unrestricted transformation above, on the other hand, cannot work;
indeed $\mathcal{D}_1$ might contain the hypothesis
$A\rightarrow B$ and hence it cannot be moved on the right. Even
worse, $\mathcal{D}_1$ may depend on hypotheses that are locally
opened, but discharged below $B$ but above $C$. Again, it is not
possible to move $\mathcal{D}_1$ on the right as naively thought,
otherwise new global hypotheses would be created. 

We overcome these barriers by our cross reductions. 
Let us highlight $\Gamma$ and $\Delta$, the hypotheses of ${\cal D}_{1}$ and
${\cal D}_{2}$ that are respectively discharged below $B$ and $A$ but above the application of $\COM$.
Assume moreover, that $A\rightarrow B$ does not occur in ${\cal
D}_{1}$ and $B\rightarrow A$ does not occur in ${\cal D}_{2}$ as
hypotheses discharged by $\COM$.
 A cross reduction transforms the deduction below left into 
the deduction below right
(if $\COM$ in the original proof discharges in each branch exactly one occurrence of the hypotheses,
and $\Gamma$ and $\Delta$ are formulas)
\[\infer[\D]{C}{\infer*{C}{\infer[\COML]{B}{\deduce{A}{\deduce{\mathcal{D}_{1}}{\Gamma}}}}&\infer*{C}{\infer[\COMR
]{A}{\deduce{B}{\deduce{\mathcal{D}_{2}}{\Delta}}}}}
\quad \quad \quad \quad
\infer[\COMB] {C}{
\infer*{ C} {\deduce{A} {\deduce{\vspace{3pt}{\cal D}_{1}}{\vspace{3pt}\infer[\COML
]{\Gamma} {\Delta} } } } &
 \infer*{C}
{\deduce{B} { \deduce{\vspace{3pt}{\cal D}_{2}}{\vspace{3pt}\infer[\COMR]{\Delta}{\Gamma} } } }
}
\]
and into the following deduction, in the general case
\[ \vcenter{ \infer[\COMB ^{3}] {C} { \infer[\COMB ^{1}] {C}
{\infer*{C} {\infer[\COML^{1}]{ B } {\deduce{A} {\deduce{\vspace{3pt}{\cal D}_{1}}{\vspace{3pt}\Gamma}}
} } & \infer*{ C} {\deduce{A} {\deduce{\vspace{3pt}{\cal D}_{1}}{\vspace{3pt}\infer=[\COML
^{3}]{\Gamma} {\Delta} } } }} && \infer[\COMB ^{2}] {C} { \infer*{C}
{\deduce{B} { \deduce{\vspace{3pt}{\cal D}_{2}}{\vspace{3pt}\infer=[\COMR ^{3}]{\Delta}{\Gamma} } } } &
\infer*{C } {\infer[\COMR ^{2}]{A } { \deduce{B} { \deduce{\vspace{3pt}{\cal D}_{2}}{\vspace{3pt}\Delta} } } } } }}
  \]
where the double bar notation stands for an application of $\COM$ between {\em sets of
 hypotheses} $\Gamma$ and $\Delta$, which means to prove from $\Gamma$
 the conjunction of the formulas of $\Gamma$, then to prove  the
 conjunction of the formulas of $\Delta$ by means of a communication
 inference and finally obtain each formula of $\Delta$ by a series of
 $\ET$ eliminations, and vice versa.

Mindless applications of the cross reductions might lead 
to dangerous loops, see e.g. Example~ \ref{example-nonterm}. To avoid them we will allow cross reductions to
be performed only when the proof is not analytic.
Thanks to this and to other restrictions, we will prove termination and thus the
Normalization Theorem.

\section{The $\lamgtitle$-Calculus}
\label{section-system}
We introduce $\lamg$, our parallel  $\lambda$-calculus for $\LC$. $\lamg$ extends the standard Curry--Howard correspondence~\cite{Wadler}
for intuitionistic natural deduction with a parallel operator that interprets the inference for the linearity axiom. 
We describe
$\lamg$-terms and their computational behavior,
proving as main result of the section the Subject Reduction Theorem, stating that the
reduction rules preserve the type.

{\footnotesize
\smallskip
\hrule
\begin{description}
\comment{We also assume that the term formation rules are applied in such a way that in each term $t$, if $t$ contains $\Wit{a}{P}{{\alpha}}$ or $\Hyp{a}{P}{{\alpha}}$ and $t$ contains $\Wit{a}{Q}{{\alpha}}$ or $\Wit{a}{Q}{{\alpha}}$,  then $\mathsf{P}=\mathsf{Q}$.}
\comment{or $\wit \beta$ (for some individual variable $\beta$) and
$A_1,\ldots, A_n$ formulas of $\Language$.}
\item[Axioms]\hspace{10 pt}
 $\begin{array}{c}    x^A: A\qquad 
\end{array}\ \ \ \ $

\item[Conjunction] 
\hspace{30 pt}
$\vcenter{\infer{  \langle u,t\rangle: A \wedge B}{u:A & t:B}} \qquad  \vcenter{\infer{u\,\pi_0: A}{u: A\wedge B}} \qquad \vcenter{\infer{u\,\pi_1: B}{u: A\wedge B}}$
\\\\

\item[Implication] 
\hspace{40 pt}
$
\vcenter{\infer{\lambda x^{A} u: A\rightarrow B}{\infer*{u:B}{[x^{A}: A]}}}
\qquad 
\vcenter{\infer{tu:B}{ t: A\IMPL B & u:A}}
$
\\\\

\comment{

\item[Disjunction Introduction]  
$\begin{array}{c} u: A\\ \hline  \inj_{0}(u): A\vee B \end{array}\ \ \ \ $ $\begin{array}{c}   u: B\\ \hline \inj_{1}(u): A\vee B \end{array}$\\\\

\item[Disjunction Elimination] 
\AxiomC{$u: A\lor B$}
\AxiomC{$[x^{A}: A]$}
\noLine
\UnaryInfC{$\vdots$}
\noLine
\UnaryInfC{$w_{1}: C$}
\AxiomC{$[y^{B}: B]$}
\noLine
\UnaryInfC{$\vdots$}
\noLine
\UnaryInfC{$w_{2}: C$}
\TrinaryInfC{$u\, [x^{A}.w_{1}, y^{B}.w_{2}]: C$}
\DisplayProof
\\
\item[Universal Quantification] 
$\begin{array}{c} u:\forall \alpha\, A\\ \hline   u m: A[m/\alpha]
\end{array} $
$\begin{array}{c}   u: A\\ \hline  \lambda \alpha\, u:
\forall \alpha A
\end{array}$\\

where $m$ is any term of  the language $\Language$ and $\alpha$ does not occur
free in the type $B$ of any free variable  $x^{B}$ of $u$.\\

\item[Existential Quantification] 
$\begin{array}{c}  u: A[m/\alpha]\\ \hline  (
m,u):
\exists
\alpha A
\end{array}\ \ \ $
\AxiomC{$\exists \alpha\, A$}
\AxiomC{$[x^{A}: A]$}
\noLine
\UnaryInfC{$\vdots$}
\noLine
\UnaryInfC{$t: C$}
\BinaryInfC{$u\, [(\alpha, x^{A}). t]: C$}
\DisplayProof
\\
where $\alpha$ is not free in $C$ nor in the type $B$ of any free variable of $t$.\\
}

\item[Linearity Axiom]\hspace{50 pt}
$\vcenter{ \infer{u \parallel_{a} v:
C}{\infer*{u:C}{[a^{\scriptscriptstyle A\rightarrow B}: A\rightarrow
B]}&\infer*{v:C}{[a^{\scriptscriptstyle B\rightarrow A}: B\rightarrow
A]}} }$
\\

\item[Ex Falso Quodlibet] 
\hspace{50 pt}$\vcenter{\infer{\Gamma\vdash  \efq{P}{u}:
P}{\Gamma \vdash u: \bot}}\qquad \text{with $P$ atomic, $P \neq \bot$.}$
\end{description}}
\hrule
\smallskip

 The table above 
 defines a type assignment for
 $\lamg$-terms, called \textbf{proof
terms} and denoted by $t, u, v \dots$, which is isomorphic to $\NGod$.
The typing rules for axioms, implication, conjunction and ex-falso-quodlibet are standard and give rise to the simply typed $\lambda$-calculus, while parallelism is introduced by the rule for the linearity axiom. 

Proof terms may contain variables $x_{0}^{A}, x_{1}^{A}, x_{2}^{A},
\ldots$ of type $A$ for every formula $A$; these variables are 
denoted as $x^{A},$ $ y^{A}, $ $ z^{A} , \ldots,$ $ a^{A}, b^{A}, c^{A}$ and
whenever the type is not important simply as $x, y, z, \ldots, a,
b$. For clarity, the variables introduced by the $\COM$ rule will be
often denoted with letters $a, b, c, \ldots$, but they are not in a
syntactic category apart. 
A variable $x^{A}$ that occurs in a term of the form $\lambda x^{A} u$ is called \textbf{$\lambda$-variable} and a variable $a$ that occurs in a term $u\parallel_{a} v$ is called \textbf{communication variable} and represents a \emph{private} communication channel between the parallel processes $u$ and $v$.

The free and bound variables of a proof term are defined as usual and
for the new term $\Ecrom{a}{u}{v}$, all the free occurrences of $a$ in
$u$ and $v$ are bound in $\Ecrom{a}{u}{v}$. In the following we
assume the standard renaming rules and alpha equivalences that are
used to avoid capture of variables in the reduction rules.

\textbf{Notation}. The connective $\rightarrow$ associates to the right and 
 \comment{so that 
 \begin{gather*}
   A_{1}\rightarrow A_{2}\rightarrow\ldots \rightarrow A_{n}
\\
   =
\\
   A_{1}\rightarrow (A_{2}\rightarrow (\ldots (A_{n-1}\rightarrow
   A_{n})\ldots))
 \end{gather*}}
by $\langle t_{1}, t_{2}, \ldots, t_{n}\rangle $ we
denote the term
 $\langle t_{1}, \langle t_{2}, \ldots \langle t_{n-1},
t_{n}\rangle\ldots \rangle\rangle$ 
and by $\proj_{i}$, for $i=0,
\ldots, n$, the sequence of projections $\pi_{1}\ldots \pi_{1}
\pi_{0}$ selecting the $(i+1)$th element of the sequence. Therefore, for every formula sequence $A_{1} , \dots ,A_{n}$ the expression $A_{1} \ET \dots \ET A_{n}$ denotes  $( A_{1} \ET ( A_{2} \ET \dots ( A_{n-1} \ET
A_{n})\dots ))$ or $\top$ if $n=0$.

Often, when $\Gamma= x_{1}: A_{1}, \ldots, x_{n}: A_{n}$ and the list
$x_{1}, \ldots, x_{n}$ includes all the free variables of a proof term
$t: A$, we shall write $\Gamma\vdash t: A$. From the logical point of
view,  
$t$ represents a natural deduction of
$A$ from the hypotheses $A_{1}, \ldots, A_{n}$. We shall write
$\LC\vdash t: A$ whenever $\vdash t: A$, and the notation means
provability of $A$ in propositional G\"odel logic. If the symbol
$\parallel$ does not occur in it, then $t$ is a \textbf{simply typed
$\lambda$-term} representing an intuitionistic deduction.

We define as usual the notion of context $\mathcal{C}[\ ]$ as the part of a proof term that surrounds a hole, represented by some fixed variable. In the expression $\mathcal{C}[u]$ we denote a particular occurrence of a subterm $u$ in the whole term $\mathcal{C}[u]$. We shall just need those particularly simple contexts which happen to be simply typed $\lambda$-terms.

\begin{definition}[Simple Contexts]\label{defi-simplec}
A \textbf{simple context} $\mathcal{C}[\ ]$ is a simply typed
$\lambda$-term with some fixed variable $[]$ occurring exactly
once. For any proof term $u$ of the same type of $[]$,
$\mathcal{C}[u]$ denotes the term obtained replacing $[]$ with $u$ in
$\mathcal{C}[\ ]$, \emph{without renaming of any bound
variable}.
\end{definition}

As an example, the expression $\mathcal{C}[\ ]:= \lambda x\, z\, ([])$ is a simple context and
 the term $\lambda x\, z\, (x\, z)$ can be written as $\mathcal{C}[xz]$.

\comment{
\begin{definition}[Parallel Contexts]\label{defi-parallelc}
Omitting parentheses, a \textbf{parallel context} $\mathcal{C}[\ ]$ is an expression of the form 
$$u_{1}\parallel_{a_{1}} u_{2}\parallel_{a_{2}}\ldots u_{i} \parallel_{a_{i}} [] \parallel_{a_{i+1}}u_{i+1}\parallel_{a_{i+2}}\ldots \parallel_{a_{n}} u_{n}$$
where $[]$ is a placeholder and $u_{1}, u_{2}, \ldots,  u_{n}$ are proof terms.
For any proof term $u$, $\mathcal{C}[u]$ denotes the replacement in $\mathcal{C}[\ ]$ of the placeholder $[]$ with $u$:
$$u_{1}\parallel_{a_{1}} u_{2}\parallel_{a_{2}}\ldots u_{i} \parallel_{a_{i}} u \parallel_{a_{i+1}}u_{i+1}\parallel_{a_{i+2}}\ldots \parallel_{a_{n}} u_{n}$$
\end{definition}}
We define below the notion of stack, corresponding to Krivine stack
\cite{Krivine} and known as \emph{continuation} because it embodies a
series of tasks that wait to be carried out. A stack represents, from
the logical perspective, a series of elimination rules; from the
$\lambda$-calculus perspective, a series of either operations
or arguments.
\begin{definition}[Stack]\label{definition-stack}
A \textbf{stack} is a sequence \mbox{$\sigma = \sigma_{1}\sigma_{2}\ldots \sigma_{n} $}
such that for every $	1\leq i\leq n$, exactly one of the following holds:
either $\sigma_{i}=t$, with $t$ proof term or $\sigma_{i}=\pi_{j}$, with $j\in\{0,1\}$.
We will denote the \emph{empty sequence} with $\epsilon$ and with
$\xi, \xi', \ldots$ the stacks of length $1$. If $t$ is a proof term,
as usual  
$t\, \sigma$ denotes the term $(((t\,
\sigma_{1})\,\sigma_{2})\ldots \sigma_{n})$.
 \end{definition}
 
We define now the notion of \emph{strong subformula}, 
which is essential for defining the reduction rules of the
$\lamg$-calculus and for proving Normalization. The technical
motivations will become clear in Sections~\ref{section-subformula}
and~\ref{section-normalization}, but the intuition is that the new types created
by cross reductions must be 
always strong subformulas of already existing types.
To define the concept of strong subformula we also need the following definition. 
\begin{definition}[Prime Formulas and Factors \cite{Krivine1}]
  A formula is said to be \textbf{prime} if it is not a
  conjunction. Every formula is 
  a conjunction of prime
  formulas, called \textbf{prime factors}. 
\end{definition}

\begin{definition}[Strong Subformula]\label{definition-strongsubf}
$B$ is said to be a \textbf{strong subformula} of a formula $A$, if $B$ is a proper subformula of some prime proper subformula  of  $A$. 
\end{definition}

Note that in the present context, prime formulas are either atomic formulas or arrow formulas, so a strong subformula of $A$ must be actually a proper subformula of an arrow proper subformula of $A$. 
The following characterization of the strong subformula relation will be often used.  
\begin{proposition}[Characterization of Strong Subformulas]\label{proposition-strongsubf}
Suppose $B$ is any {strong subformula} of $A$. Then:
\begin{itemize}
\item If $A=A_{1}\land \ldots \land A_{n}$, with $n>0$ and $A_{1}, \ldots, A_{n}$ are prime, then $B$ is a proper subformula of one among $A_{1}, \ldots, A_{n}$. 
\item If $A=C\rightarrow D$, then $B$ is a proper subformula of a prime factor of $C$ or $D$. 
\end{itemize}
\end{proposition}
\begin{proof}\mbox{}
See the Appendix.\end{proof}

 \begin{definition}[Multiple Substitution]\label{defi-multsubst}
Let $u$ be a proof term, $\sq{x}=x_{0}^{A_{0}}, \ldots, x_{n}^{A_{n}}$ a sequence of  variables and $v: A_{0}\land \ldots \land A_{n}$. The substitution 
$u^{v/ \sq x}:=u[v\,\proj_{0}/x_{0}^{A_{0}} \ldots \,v\,\proj_{n}/x_{n}^{A_{n}}]$
replaces each variable $x_{i}^{A_{i}}$ of
any term $u$ with the $i$th projection of $v$.
\end{definition}

We now seek a measure for determining how complex the communication
channel $a$ of a term $u\parallel_{a} v$ is. Logic will be our
guide.  
First, it makes sense to consider the types $B, C$ such
that $a$ occurs with type $B\rightarrow C$ in $u$ and thus with type
$C\rightarrow B$ in $v$. Moreover, assume $u\parallel_{a} v$ has type
$A$ and its free variables are $x_{1}^{A_{1}}, \ldots,
x_{n}^{A_{n}}$. The Subformula Property tells us that, no matter what
our notion of computation will turn out to be, when the computation is
done, no object of type more complex than the types of the inputs and
the output should appear.  Hence, if the prime factors of the types
$B$ and $C$ are not subformulas of $A_{1}, \ldots, A_{n}, A$, then
these prime factors should be taken into account in the complexity
measure we are looking for. The actual definition is the following.

\begin{definition}[Communication Complexity]\label{definition-comcomplexity}
Let $u\parallel_{a} v: A$ a proof term with free variables $x_{1}^{A_{1}}, \ldots, x_{n}^{A_{n}}$.
Assume that  $a^{B\rightarrow C}$ occurs in $u$ and thus $a^{C\rightarrow B}$ in $v$. 
\begin{itemize}
\item The pair $B, C$ is called the \textbf{communication kind}  of $a$. 
\item  The \textbf{communication complexity} of  $a$ is the maximum  among $0$ and the numbers of symbols of the prime factors of $B$ or $C$ that are neither proper subformulas of $A$ nor strong subformulas of 
one among $A_{1}, \ldots, A_{n}$.
\end{itemize}
\end{definition}

We  
explain now the basic reduction rules for the proof terms of $\lamg$, which are given in Figure \ref{fig:red}.
As usual, we also have the reduction scheme: $\mathcal{E}[t]\mapsto \mathcal{E}[u]$, whenever $t\mapsto u$ and for any context $\mathcal{E}$. With $\mapsto^{*}$ we shall denote the reflexive and transitive closure of the one-step reduction $\mapsto$.

 \textbf{Intuitionistic Reductions}. These are the very familiar
computational rules for the simply typed $\lambda$-calculus,
representing the operations of applying a function 
and taking a component of a pair \cite{Girard}. From the logical point
of view, they are the standard Prawitz reductions \cite{Prawitz} for
$\NJ$.

  \textbf{Cross Reductions}. The reduction rules for $\COM$ 
  model a communication mechanism between parallel processes. In order to apply a cross reduction to a term 
 $$ \mathcal{C}[a\, u]\parallel_{a} \mathcal{D}[a\, v]$$
several conditions have to be met. These conditions are both natural \emph{and} needed for the termination of computations.\\
\emph{First}, we require 
the communication complexity of $a$ 
to be greater than $0$; again,
this is 
a warning that the Subformula Property does not hold. Here we
use a logical property as a \emph{computational criterion} for
answering the question: when should computation stop? An answer is  crucial here, because, as shown in Example \ref{example-nonterm}, unrestricted cross reductions do not always terminate.
In $\lambda$-calculi the Subformula
Property fares pretty well as a stopping criterion. In a sense, it
detects all the \emph{essential} operations that really have to be
done.  
For example, in simply typed
$\lambda$-calculus, a closed term that has the Subformula Property
must be a \emph{value}, that is, of the form $\lambda x\, u$,
or $\langle u, v \rangle$.
Indeed a closed term which is a
not a value, must be of the form $h\, \sigma$, for some stack $\sigma$
(see Definition \ref{definition-stack}), where $h$ is a redex
$(\lambda y\, u)t$ or $\langle u, v\rangle\, \pi_{i}$; but $(\lambda
y\, u)$ and $\langle u, v\rangle$ would have a more complex type than
the type of the whole term, contradicting the Subformula Property.  
\\
\emph{Second}, we require 
$\mathcal{C}[a\, u], \mathcal{D}[a\, v]$
to be normal simply typed $\lambda$-terms. Simply typed
$\lambda$-terms, because they are easier to execute in parallel; normal, because we want their computations to go on until they are really stuck and communication is  unavoidable.
\emph{Third}, we require 
the variable $a$ 
to be as rightmost as
possible and that is
 needed for logical soundness: how could otherwise
the term $u$ be moved to the right, e.g., if it contains $a$?

  Assuming that all the conditions above are satisfied, we can now start to explain the cross reduction
\begin{small}
\[\mathcal{C}[a\, u]\parallel_{a} \mathcal{D}[a\, v] \mapsto  (  \mathcal{D}[u^{b\langle \sq{z}\rangle / \sq{y}}] \parallel_{a} \mathcal{C}[a\, u] ) \parallel_{b} (\mathcal{C}[v^{b\langle \sq{y}\rangle / \sq{z}}]\parallel_{a} \mathcal{D}[a\, v])\]
\end{small}
Here, the communication channel $a$ has been activated, because the processes $\mathcal{C}$ and $\mathcal{D}$ are synchronized and ready to transfer respectively $u$ and $v$. The parallel operator $\parallel_{a}$ let the two occurrences of $a$ communicate: the term $u$ travels to the right in order to replace $a\, v$ and $v$ travels to the left in order to replace $a\, u$.  
If $u$ and $v$ were data, like numbers or constants, everything would be simple and they could be sent as they are; but in general, this is not possible.
The problem is that the free variables $\sq{y}$ of $u$ which are bound
in $\mathcal{C}[a\, u]$ by some $\lambda$ cannot be permitted to
become free; otherwise, the connection between the binders $\lambda\,
\sq{y}$ and the occurrences of the variables $\sq{y}$ would be lost
and they could be no more replaced by actual values when the inputs
for the $\lambda\, \sq{y}$ are available. Symmetrically, the variables
$\sq{z}$  cannot become free. For example, we could have $u=u'\,
\sq{y}$ and $v=v'\,\sq{z}$ and\begin{small}
  \[\mathcal{C}[a\, u]=w_{1}\, (\lambda
    \sq{y}\, a\, (u'\, \sq{y}))\qquad\ \mathcal{D}[a\, v]=w_{2}\,
    (\lambda \sq{z}\, a\, (v'\, \sq{z}))\]
\end{small}and the transformation \begin{small} 
    $w_{1}\, (\lambda \sq{y}\, a\, (u'\, \sq{y}))\parallel_{a} w_{2}\,
    (\lambda \sq{z}\, a\, (v'\, \sq{z}))$
$\mapsto$ $w_{1}\, (\lambda
    \sq{y}\, v'\, \sq{z} )\parallel_{a} w_{2}\, (\lambda \sq{z}\, u'\,
    \sq{y})$
\end{small}
would just be wrong: the term $v'\, \sq{z}$ will never get back actual values for the variables $\sq{z}$ when they will become available.

\noindent These issues are typical of \emph{process migration},
and can be solved by the concepts of
\emph{code mobility}~\cite{codemobil} and \emph{closure}~\cite{landin}. Informally, code mobility is defined as the
capability to dynamically change the bindings between code fragments
and the locations where they are executed. 
Indeed,  in order to be executed, a piece of code needs a computational environment and its resources, 
like data, program counters or global variables.
In our case the
contexts $\mathcal{C}[\ ]$ and $\mathcal{D}[\ ]$ are the computational
environments or \emph{closures} of the processes $u$ and $v$ and the variables $\sq{y},
\sq{z}$ are the resources they need. Now, moving a process outside its
environment always requires extreme care: the bindings between a
process and the environment resources must be preserved. This is the
task of the \emph{migration mechanisms}, which allow a migrating
process to resume correctly its execution in the new location. Our
migration mechanism creates a new communication channel $b$ between
the programs that have been exchanged. Here we see the code fragments
$u$ and $v$, with their original bindings to the global variables
$\sq{y}$ and $\sq{z}$.
\begin{center}
\vspace{-10pt}
\includegraphics[width=0.25\textwidth]{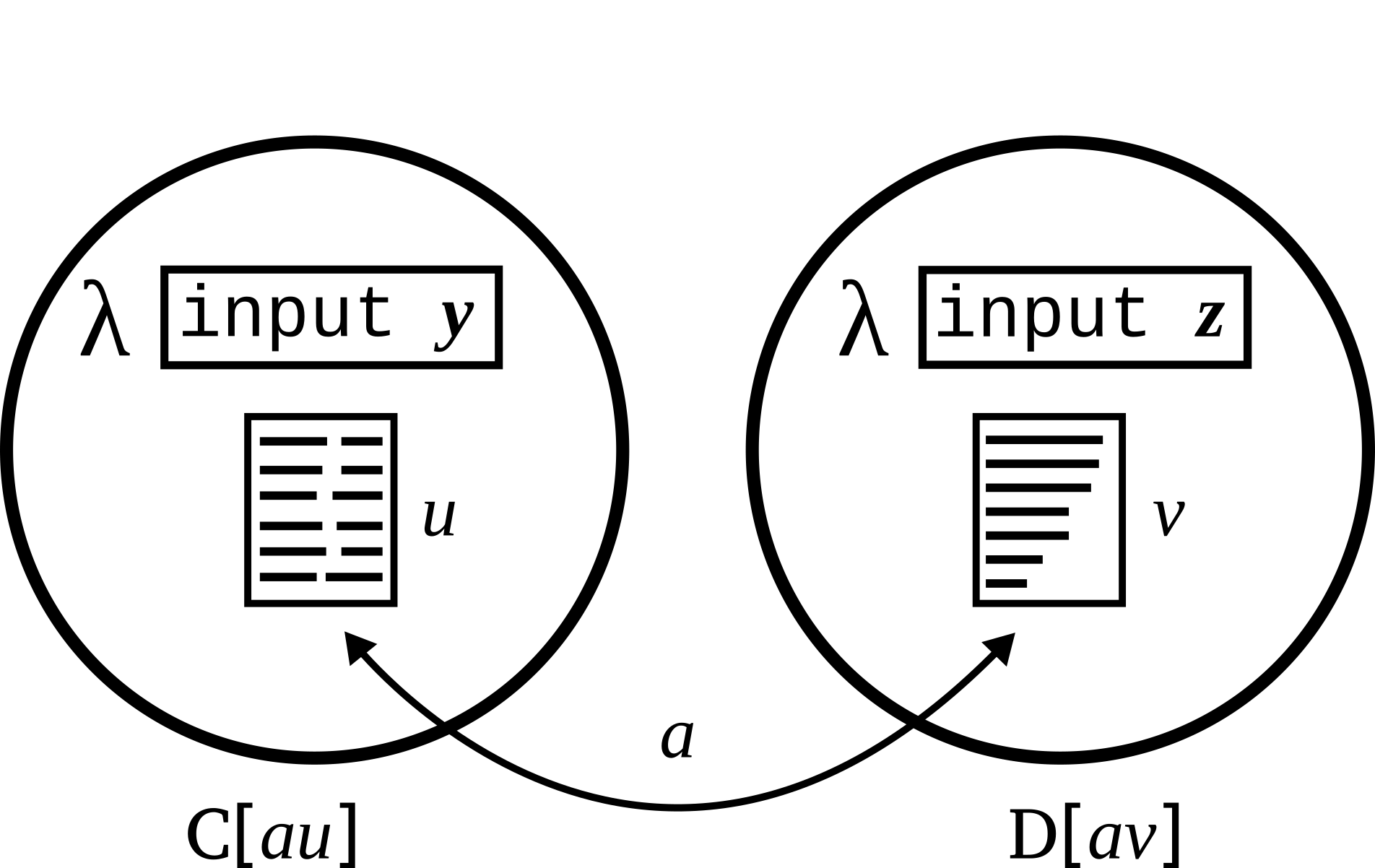}
\end{center}
The change of variables $u^{b\langle \sq{z}\rangle / \sq{y}}$ and $ v^{b\langle \sq{y}\rangle / \sq{z}}$ has the effect of reconnecting $u$ and $v$ to their old inputs:
\begin{center}
\vspace{-10pt}
\includegraphics[width=0.25\textwidth]{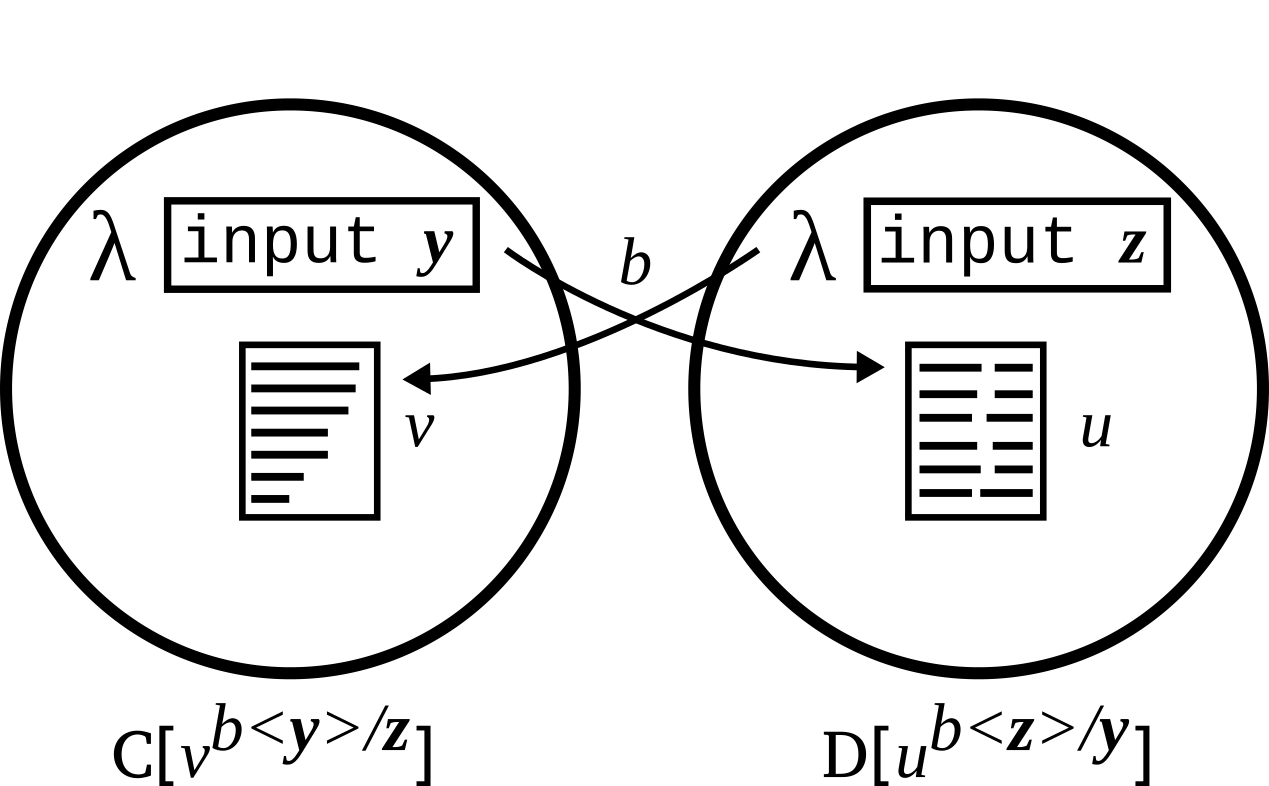}
\end{center}
In this way, when they will become available, the data $\sq{y}$ will be sent to $u$  and the data $\sq{z}$ will be sent to $v$ through the channel $b$. 
Note that in the result of the cross reduction
the processes $\mathcal{C}[a\, u]$ and $\mathcal{D}[a\, v]$ are
\emph{cloned}, because their code fragments can be needed again.
Thus $a$ behaves as a \emph{replicated input} and \emph{replicated output channel}. 
E.g., in \cite{CP2010},
replicated input is coded by the bang operator of linear logic:
$$x\langle y\rangle. Q\, |\, !x(z).P \mapsto Q\, |\, P[y/z]\, |\, !x(z).P$$ 
With symmetrical message passing and a ``!" also in front of $x\langle y\rangle. Q$, one would obtain a version of our cross
reduction. Finally, as detailed in Ex.~\ref{ex:pi_calc}, whenever $u$
and $v$ are closed terms the cross reduction is simpler and only
maintains the first two of the four processes produced in the general
case.

\begin{example}[\textbf{$\parallel_{a}$ in $\lamgtitle$ and $\mid$ in
the $\pi$-calculus}]
\label{Pi} A private channel $u\parallel_{a} v$ is rendered in the
$\pi$-calculus~\cite{Milner, sangiorgiwalker2003} by the restriction
operator $\nu$, as $\nu a\, (u\ |\ v)$.  
Recall that the $\pi$-calculus term $u \mid v$ represents two processes that run
in parallel. The corresponding $\lamg$ term $\langle e, u \rangle \parallel_{e} \langle e, v
\rangle$ is defined using a fresh channel $e$ with communication kind $A, A$.
As no cross reduction outside $u$ and $v$ can be applied, the whole term reduces neither to $\langle e, u \rangle$ nor to $\langle e, v\rangle$, so that $u$ and $v$ can run in parallel.
\end{example}

\begin{example}\label{example-nonterm}
Let $y$ and $z$ be bound variables  occurring in the normal terms $\mathcal{C}[a\,y]$ and $\mathcal{D}[a\, z]$. 
Without the condition on the communication complexity $c$ of $a$,
a loop could be generated:\begin{small}
\[\mathcal{C}[a\, y]\parallel_{a} \mathcal{D}[a\, z] \mapsto (\mathcal{D}[y^{b\langle {z}\rangle / {y}}]
\parallel_{a} \mathcal{C}[a\, y] ) \parallel_{b} (\mathcal{C}[z^{b\langle {y}\rangle / {z}}]\parallel_{a}
\mathcal{D}[a\, z])\]
\[=(\mathcal{D}[b\,{z}] \parallel_{a} \mathcal{C}[a\,
y]) \parallel_{b} (\mathcal{C}[b\, {y}]\parallel_{a} \mathcal{D}[a\,
z])\mapsto^{*}   \mathcal{D}[b\,{z}] \parallel_{b} \mathcal{C}[b\,
{y}]\]\end{small}In Sec.~\ref{section-normalization} we show that if $c > 0$, this
reduction sequence would terminate. What is then happening here?
Intuitively, $\mathcal{C}[a\,y]$ and $\mathcal{D}[a\, z]$ are normal
simply typed $\lambda$-terms, which forces $c=0$. 
\end{example}

 \textbf{Permutation Reductions}. They regulate the interaction
between parallel operators and the other computational constructs. The
first four reductions are the Prawitz-style permutation rules
\cite{Prawitz} between parallel operators and eliminations.  
We also add two other groups of reductions: three permutations between
parallel operators and introductions, two permutations between
parallel operators themselves. The first group will be needed to
rewrite any proof term into a parallel composition of simply typed
$\lambda$-terms (Proposition \ref{proposition-parallelform}).
The second group is needed to address the \emph{scope extrusion} issue of private channels \cite{Milner}. We point out that a parallel operator $\parallel_{a}$ is allowed to commute with other parallel operators only when it is strictly
necessary, that is, when the communication complexity of $a$ is
greater than $0$ and thus signaling a violation of the Subformula
Property. 
\begin{example}[\textbf{Scope extrusion (and $\pi$-calculus)}]
\label{SE}
As example of scope extrusion, let us consider the term
\[(v\parallel_{a} \mathcal{C}[b\, a] )\parallel_{b} w 
\]
Here the process $\mathcal{C}[b\, a]$ wishes to send the channel $a$ to $w$ along the channel $b$, but this is not possible being the channel $a$ private. This issue is solved in the $\pi$-calculus using the congruence $\nu a (P\, |\, Q)\, |\, R \equiv \nu a (P\, |\, Q\, |\, R) $,
 provided that  $a$ does not occur in $R$, condition that can always be satisfied  by $\alpha$-conversion. G\"odel logic offers and actually forces a different solution, which is not just permuting $w$ inward but also duplicating it:
\[( v \parallel_{a} \mathcal{C}[b\, a])\parallel_{b} w \mapsto (v \parallel_{b} w)\parallel_{a} (\mathcal{C}[b\, a]\parallel_{b} w)
\]
After this reduction $\mathcal{C}[b\, a]$ can send $a$ to $w$. If $a$
does not occur in $v$, we have a further simplification step:
$$(v \parallel_{b} w)\parallel_{a} (\mathcal{C}[b\, a]\parallel_{b}
w)\mapsto v \parallel_{a} (\mathcal{C}[b\, a]\parallel_{b} w)$$
obtaining associativity of composition as in $\pi$-calculus. However,
if $b$ occurs in $v$, this last reduction step is not possible and we
keep both copies of $w$. It is indeed natural to allow both $v$ and
$\mathcal{C}[b\,a]$ to communicate with $w$.
\end{example}

Everything works as expected: the reductions steps in
Fig.~\ref{fig:red} preserve the type at the level of proof terms,
i.e., they correspond to logically sound proof transformations. Indeed

\begin{theorem}[Subject Reduction]\label{subjectred}
If $t : A$ and $t \mapsto u$, then $u : A$ and all the free variables of $u$ appear among those of $t$.
\end{theorem} 
\begin{proof} 
It is enough to prove the theorem for basic reductions: if $ t : A$
and $t \mapsto u$, then $u : A$. The proof that the intuitionistic
reductions and the permutation rules preserve the type is completely
standard 
(full proof in the Appendix). 
Cross reductions require straightforward considerations as well. Indeed 
suppose
\begin{small}
  \begin{gather*}
    \mathcal{C}[a^{\scriptscriptstyle A\rightarrow B}\,
    u]\parallel_{a} \mathcal{D}[a^{\scriptscriptstyle B\rightarrow
      A}\, v]\ \\ \mapsto
    \\
    ( \mathcal{D}[u^{b^{\scriptscriptstyle D\rightarrow C}\langle
      \sq{z}\rangle / \sq{y}}]
    \parallel_{a} \mathcal{C}[a^{\scriptscriptstyle A\rightarrow B}\,
    u] ) \parallel_{b} (\mathcal{C}[v^{b^{\scriptscriptstyle
        C\rightarrow D}\langle \sq{y}\rangle / \sq{z}}]\parallel_{a}
    \mathcal{D}[a^{\scriptscriptstyle B\rightarrow A}\, v])
  \end{gather*}
\end{small}Since  
$
   \langle \sq{y}\rangle: C := C_{0}\land \ldots\land C_{n} \; \mbox{and} \;
   \langle \sq{z}\rangle: D := D_{0}\land \ldots\land D_{m}
$, $b^{\scriptscriptstyle
        D\rightarrow C}\langle \sq{z}\rangle$ and $b^{\scriptscriptstyle
        C\rightarrow D}\langle \sq{y}\rangle$ are correct
      terms. Therefore $u^{b^{\scriptscriptstyle D\rightarrow C}\langle
      \sq{z}\rangle / \sq{y}}$ and  $v^{b^{\scriptscriptstyle
        C\rightarrow D}\langle \sq{y}\rangle / \sq{z}}$, by Definition ~\ref{defi-multsubst},
    are correct as well.
The assumptions are that
$\sq{y}=y_{0}^{C_{0}}, \ldots, y_{n}^{C_{n}}$ is the sequence of the
free variables of $u$ which are bound in $\mathcal{C}[a^{\scriptscriptstyle A\rightarrow B}u]$,
$\sq{z}=z_{0}^{D_{0}}, \ldots, z_{m}^{D_{m}}$ is the sequence of the
free variables of $v$ which are bound in $\mathcal{D}[a^{\scriptscriptstyle B\rightarrow A}v]$, $a$ does
not occur neither in $u$ nor in $v$ and $b$ is fresh.\
Therefore, by construction all the variables $\sq{z}$ are bound in
$\mathcal{D}[u^{b^{\scriptscriptstyle D\rightarrow C}\langle
\sq{z}\rangle / \sq{y}}]$ and all the variables $\sq{y}$ are bound in
$\mathcal{C}[v^{b^{\scriptscriptstyle C\rightarrow D}\langle
\sq{y}\rangle / \sq{z}}]$. Hence, no new free variable is created.\end{proof}

 \begin{figure*}[!htb]
\footnotesize{
\hrule
\smallskip
\begin{description}

\item \textbf{Intuitionistic Reductions} \hspace{55pt}
$(\lambda x^{\scriptscriptstyle A}\, u)t\mapsto
u[t/x^{\scriptscriptstyle A}] \quad \mbox{and} \quad
 \pair{u_0}{u_1}\,\pi_{i}\mapsto u_i, \mbox{ for $i=0,1$}$
\bigskip

\item \textbf{Permutation Reductions} \hspace{57pt} $(\Ecrom{a}{u}{v}) w \mapsto \Ecrom{a}{uw}{vw},\mbox{ if $a$ does
      not occur free in $w$} $
\[w(\Ecrom{a}{u}{v})  \mapsto \Ecrom{a}{wu}{wv},\mbox{ if $a$ does not occur free in $w$} \]
\[\efq{P}{\Ecrom{a}{w_{1}}{w_{2}}}  \mapsto \Ecrom{a}{\efq{P}{w_{1}}}{\efq{P}{w_{2}}}\]
\[(\Ecrom{a}{u}{v})\,\pi_{i}  \mapsto \Ecrom{a}{u\,\pi_{i}}{v\,\pi_{i}} \]

\[\lambda x^{\scriptscriptstyle A}\,(\Ecrom{a}{u}{v})  \mapsto
  \Ecrom{a}{\lambda x^{\scriptscriptstyle A}\,u}{\lambda
    x^{\scriptscriptstyle A}\, v} \]
\[\langle u \parallel_{a} v,\, w\rangle \mapsto \langle u, w\rangle \parallel_{a} \langle v, w\rangle, \mbox{ if $a$ does not occur free in $w$}\]
\[\langle w, \,u \parallel_{a} v\rangle \mapsto \langle w, u\rangle \parallel_{a} \langle w, v\rangle, \mbox{ if $a$ does not occur free in $w$}\]

\[(u\parallel_{a} v)\parallel_{b} w \mapsto (u\parallel_{b} w)\parallel_{a} (v\parallel_{b} w),\mbox{ if the communication complexity of $b$ is greater than  $0$}
\]
\[w \parallel_{b} (u\parallel_{a} v) \mapsto (w\parallel_{b} u)\parallel_{a} (w\parallel_{b} v),\mbox{ if the communication complexity of $b$ is greater than  $0$}
\]
\vspace{0.5pt}

\item \textbf{Cross Reductions}  \hspace{30pt}
$u\parallel_{a}v \mapsto u, \mbox{ if $a$ does not occur in $u$} \quad \mbox{and} \quad
u\parallel_{a}v \mapsto v, \mbox{ if $a$ does not occur in $v$}$
\vspace{3pt}
\[\mathcal{C}[a^{\scriptscriptstyle A \IMPL B}\, u]\parallel_{a} \mathcal{D}[a^{\scriptscriptstyle B \IMPL A}\, v]\ \mapsto\
(\mathcal{D}[u^{b^{\scriptscriptstyle C \IMPL D}\langle \sq{z}\rangle
/ \sq{y}}]  \parallel_{a} \mathcal{C}[a^{\scriptscriptstyle A \IMPL B}\, u] ) \, \parallel_{b}\, (\mathcal{C}[v^{b^{\scriptscriptstyle D \IMPL C}\langle \sq{y}\rangle /
\sq{z}}]\parallel_{a} \mathcal{D}[a^{\scriptscriptstyle B \IMPL A}\, v])\]
where $\mathcal{C}[a\, u], \mathcal{D}[a\, v]$ are normal simply typed $\lambda$-terms and $\mathcal{C}, \mathcal{D}$ simple contexts;
$\sq{y}$  is the sequence of the free variables of $u$ which are bound in $\mathcal{C}[a\,u]$;
$\sq{z}$ is the sequence of the free variables of $v$ which are bound
in $\mathcal{D}[a\,v]$; $C$ and $D$ are the conjunctions of the types
of the variables in $\sq{z}$ and $\sq{y}$, respectively; the displayed occurrences of $a$ are the rightmost both in $\mathcal{C}[a\,u]$ and in $\mathcal{D}[a\,v ]$; $b$ is fresh;
 and the communication complexity of $a$ is greater than $0$

\comment{
\item[Reduction Rules for $\D$]
\[\left(u_{0}\parallel_{a_{1}}\ldots \parallel_{a_{i}} a\, t\,\sigma \parallel_{a_{i+1}}\ldots \parallel_{a_{n}} u_{n}\right)\, \parallel_{a}\, v\ \mapsto\  \left(u_{0}\parallel_{a_{1}}\ldots \parallel_{a_{i}} v[(\lambda y\, t)/a] \parallel_{a_{i+1}}\ldots \parallel_{a_{n}} u_{n}\right)\, \parallel_{a}\, v\]
\[v\, \parallel_{a}\, \left(u_{0}\parallel_{a_{1}}\ldots \parallel_{a_{i}} a\, t\,\sigma \parallel_{a_{i+1}}\ldots \parallel_{a_{n}} u_{n}\right)\ \mapsto\ v\, \parallel_{a}\, \left(u_{0}\parallel_{a_{1}}\ldots \parallel_{a_{i}} v[(\lambda y\, t)/a] \parallel_{a_{i+1}}\ldots \parallel_{a_{n}} u_{n}\right)\]
\[\mbox{ whenever $a\, t\, \sigma$ is the leftmost head redex of the whole term  and}\]
\[\mbox{${y}$ is a variable not occurring in $t$. 
}\]}
\end{description}}
\hrule
\caption{Basic Reduction Rules for $\lamg$
}\label{fig:red}
\end{figure*}

 \begin{definition}[Normal Forms and Normalizable Terms]\mbox{}
 \begin{itemize}
\item  A \textbf{redex} is a term $u$ such that $u\mapsto v$ for some $v$ and basic reduction of Figure \ref{fig:red}. A term $t$ is called a \textbf{normal form} or, simply, \textbf{normal}, if there is no $t'$ such that $t\mapsto t'$. We define $\nf$ to be the set of normal $\lamg$-terms. 
\item  
A sequence, finite or infinite, of proof terms
$u_1,u_2,\ldots,u_n,\ldots$ is said to be a reduction of $t$, if
$t=u_1$, and for all  $i$, $u_i \mapsto
u_{i+1}$.
 A proof term $u$ of $\lamg$ is  \textbf{normalizable} if there is a finite reduction of $u$ whose last term is a normal form. 
\end{itemize}
\end{definition}

\begin{definition}[Parallel Form]\label{definition-head}
A term $t$ is a \textbf{parallel form}
  whenever, removing the parentheses, it can be written as
$$t = t_{1}\parallel_{a_{1}} t_{2}\parallel_{a_{2}}\ldots \parallel_{a_{n}} t_{n+1}$$
where each $t_{i}$, for $1\leq i\leq n+1$, is a simply typed $\lambda$-term.
  \end{definition}

\section{ 
The Subformula Property}\label{section-subformula}

We show that normal $\lamg$-terms satisfy the important
Subformula Property (Theorem \ref{theorem-subformula}). This, in turn,
implies that our Curry--Howard correspondence for $\lamg$ is meaningful
from the logical perspective and produces analytic $\NGod$ proofs.

We start by establishing an elementary property of simply typed $\lambda$-terms, which will turn out to be crucial for our normalization
proof. It ensures that every bound hypothesis appearing in a normal
intuitionistic proof is a strong subformula of one the premises or a
proper subformula of the conclusion. This property sheds light on the
complexity of cross reductions, because it implies that the new
formulas introduced by these operations are always smaller
than the local premises.

\begin{proposition}[Bound Hypothesis Property]\label{proposition-boundhyp}
Suppose
$$x_{1}^{A_{1}}, \ldots, x_{n}^{A_{n}}\vdash t: A$$
$t\in\nf$ is a simply typed $\lambda$-term and  $z: B$ a variable occurring bound in $t$. 
Then one of the following holds:
\begin{enumerate}
\item $B$ is a proper subformula of a prime factor of $A$.
\item $B$ is a strong subformula of one among $A_{1},\ldots, A_{n}$.

\end{enumerate}
\end{proposition}

\begin{proof}
By induction on $t$. 
See the Appendix.
\end{proof}

The next proposition says that each occurrence of any hypothesis
of a normal intuitionistic proof must be followed by an
elimination rule, whenever the hypothesis is neither $\bot$ nor a
subformula of the conclusion nor a proper subformula of some other
premise.

\begin{proposition}\label{proposition-app}
Let $t\in \nf$ be a simply typed $\lambda$-term and 
$$x_{1}^{A_{1}}, \ldots, x_{n}^{A_{n}}, z^{B}\vdash t: A$$
One of the following holds:

\begin{enumerate}

\item Every occurrence of $z^{B}$ in $t$ is of the form $z^{B}\, \xi$ for some proof term or projection $\xi$.

\item $B=\bot$ or $B$ is a subformula of  $A$ or a  proper subformula of one among the formulas $A_{1}, \ldots, A_{n}$.
\end{enumerate}
\end{proposition}
\begin{proof}
  Easy structural induction on the term. 
See the Appendix.\end{proof}

\begin{proposition}[Parallel Normal Form Property]
\label{proposition-parallelform} 
If $t\in
\nf$ is a $\lamg$-term, then it is in parallel form.
\end{proposition}
\begin{proof}
      Easy structural induction on $t$ using the permutation
      reductions. 
See the Appendix.\end{proof}

We finally prove the Subformula Property: a normal proof does not
contain concepts that do not already appear in the premises or in the
conclusion.

\begin{theorem}[Subformula Property]\label{theorem-subformula}
Suppose
$$x_{1}^{A_{1}}, \ldots, x_{n}^{A_{n}}\vdash t: A \quad \mbox{and} \quad
t\in \nf. \quad \mbox{Then}:$$ 
\begin{enumerate}
\item
For each communication variable $a$ occurring bound in  $t$ and with communication kind $B, C$, the prime factors of $B$ and $C$ are proper subformulas of  
$A_{1}, \ldots, A_{n}, A$. 
\item The type of any subterm of $t$ which is not a bound communication variable is either a subformula or a conjunction of subformulas of the formulas $A_{1}, \ldots, A_{n}, A$. 
\end{enumerate}

\end{theorem}
\begin{proof} 
We proceed by induction on $t$. 
By Proposition \ref{proposition-parallelform}
$t = t_{1}\parallel_{a_{1}} t_{2}\parallel_{a_{2}}\ldots \parallel_{a_{n}} t_{n+1}$
and each $t_{i}$, for $1\leq i\leq n+1$, is a simply typed
$\lambda$-term. We only show the case $t= u_{1}\parallel_{b} u_{2}$. 
Let $C, D$ be the communication
kind of $b$, we first show that the communication complexity of $b$ is $0$.
  We reason by contradiction and assume that it is greater than $0$. 
 $u_{1}$ and $u_{2}$ are either
simply typed $\lambda$-terms or of the form $v\parallel_{c} w$. The
second case is not possible, otherwise a permutation reduction could
be applied to $t\in \nf$. Thus $u_{1}$ and $u_{2}$ are simply typed
$\lambda$-terms. Since the communication complexity of $b$ is greater
than $0$, the types $C\rightarrow D$ and $D\rightarrow C$ are not
subformulas of $A_{1}, \ldots, A_{n}, A$. By Prop.~\ref{proposition-app}, every occurrence of $b^{C\rightarrow D}$ in
$u_{1}$ is of the form $b^{C\rightarrow D} v$ and every occurrence of
$b^{D\rightarrow C}$ in $u_{2}$ is of the form $b^{D\rightarrow C} w$.
Hence, we can write
$$u_{1}=\mathcal{C}[b^{C\rightarrow D} v] \qquad u_{2}=\mathcal{D}[b^{D\rightarrow C} w]$$
where $\mathcal{C}, \mathcal{D}$ are simple contexts and $b$ is rightmost. Hence a cross reduction of $t$ can be performed,
which contradicts the fact that $t\in\nf$.  Since we have established
that the communication complexity of $b$ is $0$, the prime factors of
$C$ and $D$ must be proper subformulas of $A_{1}, \ldots, A_{n},
A$. Now, by induction hypothesis applied to $u_{1}: A$ and $u_{2}: A$,
for each communication variable $a^{F\rightarrow G}$ occurring bound
in $t$, the prime factors of $F$ and $G$ are proper subformulas of the
formulas $ A_{1},  \ldots, A_{n}, A, C\rightarrow D, D\rightarrow C$
and thus of the formulas $A_{1}, \ldots, A_{n}, A$; moreover, the
type of any subterm of $u_{1}$ or $u_{2}$ which is not a communication
variable is either a subformula or a conjunction of subformulas of the
formulas $ A_{1}, \ldots, A_{n}, C\rightarrow D, D\rightarrow C$ and
thus of $A_{1},  \ldots, A_{n}, A$.\end{proof}

\begin{remark} Our statement of the Subformula Property is slightly
different from the usual one. However the latter can be easily
recovered using the communication rule
$(com_{end})$ of Section~\ref{section-ND} and additional reduction
rules.
As the resulting derivations would be isomorphic but more complicated,
we prefer the current statement.
 \end{remark}

\section{The Normalization Theorem}\label{section-normalization}
 
Our goal is to prove the Normalization Theorem for $\lamg$:
every proof term of $\lamg$ reduces in a finite number of steps to a
normal form.  By Subject Reduction, this implies that $\NGod$ proofs
normalize. We shall define a reduction strategy for terms of $\lamg$:
a recipe for selecting, in any given term, the subterm to which apply
one of our basic reductions. We remark that the permutations between
communications have been adopted to simplify the normalization proof,
but at the same time, they undermine strong normalization, because
they enable silly loops, like in cut-elimination for sequent
calculi. Further restrictions of the permutations might be enough to
prove strong normalization, but we leave this as an open problem.

The idea behind our normalization strategy is to employ a suitable
complexity measure for terms $u\parallel_{a} v$ and, each time a
reduction has to be performed, to choose the term of maximal
complexity. Since cross reductions can be applied as long as there is
a violation of the Subformula Property, the natural approach is to
define the complexity measure as a function of some fixed set of
formulas, representing the formulas that can be safely used without
violating the Subformula Property.

\begin{definition}[Complexity of Parallel
  Terms]\label{defi-acomplexity}
Let $\mathcal{A}$ be a finite set of formulas. The $\mathcal{A}$-\textbf{complexity} of the term $u\parallel_{a} v$ 
is the  sequence $(c, d, l, o)$ of natural numbers, where:
\begin{enumerate}
\item 
  if the communication kind of $a$ is $B, C$, then $c$ is the maximum
among $0$ and the number of symbols of the prime factors of $B$ or
$C$ that are not subformulas of some formula in $\mathcal{A}$;
\item $d$ is the number of occurrences of $\parallel$ in
$u$ and $v$;
\item $l$ is the sum of the lengths of the intuitionistic reductions
of $u$ and $v$ to reach intuitionistic normal form;
\item $o$ is the number of occurrences of $a$ in $u$ and $v$.
\end{enumerate}
The $\mathcal{A}$-\textbf{communication-complexity} of  $u\parallel_{a} v$  is $c$.


\end{definition}
For clarity, we define the recursive normalization algorithm that
represents the constructive content of the proofs of
Prop.~\ref{proposition-normpar} and \ref{proposition-normcomp}, which
are used to prove the Normalization Theorem.  Essentially, our master
reduction strategy consists in iterating the basic reduction relation
$\succ$ defined below, whose goal is to permute the smallest redex
$u\parallel_{a} v$ of maximal complexity until $u$ and $v$ are simply
typed $\lambda$-terms, then normalize them and finally apply the cross
reductions.


\begin{definition}[Side Reduction Strategy]\label{defi-redstrategy}
Let $t: A$ be a term with free variables $x_{1}^{A_{1}},\ldots,
x_{n}^{A_{n}}$ and $\mathcal{A}$ be the set of the proper subformulas
of $A$ and the strong subformulas of the formulas $A_{1}, \ldots,
A_{n}$.  Let $u\parallel_{a} v$ be the \emph{smallest subterm} of $t$, if
any, among those of \emph{maximal} $\mathcal{A}$-complexity and let
$(c, d, l, o)$ be its $\mathcal{A}$-complexity.  We write\[t\succ
t'\]whenever $t'$ has been obtained from $t$ by applying to
$u\parallel_{a} v$:
\begin{enumerate}
\item a permutation reduction
$$(u_{1}\parallel_{b} u_{2})\parallel_{a} v \mapsto (u_{1}\parallel_{a} v)\parallel_{b} (u_{2}\parallel_{a} v)$$ $$u \parallel_{a} (v_{1}\parallel_{b} v_{2}) \mapsto (u\parallel_{a} v_{1})\parallel_{b} (u\parallel_{a} v_{2})$$
 if $d>0$ and $u=u_{1}\parallel_{b}u_{2}$ or
$v=v_{1}\parallel_{b}v_{2}$;
\item a sequence of intuitionistic reductions normalizing both $u$ and
$v$, if $d=0$ and $l>0$;
\item a cross reduction if $d=l=0$ and $c>0$, immediately followed by
intuitionistic reductions normalizing the newly generated simply typed
$\lambda$-terms and, if possible, by applications of the cross
reductions $u_{1} \parallel_{b} v_{1} \mapsto u_{1}$ and
$u_{1} \parallel_{b} v_{1} \mapsto v_{1}$ to the whole term.
\item a cross reduction $u \parallel_{a} v \mapsto u$ and
$u \parallel_{a} v \mapsto v$ if $d=l=c=0$.
\end{enumerate}
\end{definition}

\begin{definition}[Master Reduction Strategy]\label{defi-mastredstrategy}
We define a normalization algorithm $\nor{N}(t)$ taking as input a
typed term $t$ and producing a term $t'$ such that $t\mapsto^{*}
t'$.
Assume that the free variables of $t$ are
$x_{1}^{A_{1}},\ldots, x_{n}^{A_{n}}$ and let $\mathcal{A}$ be the set
of the proper subformulas of $A$ and the strong subformulas of the
formulas $A_{1}, \ldots, A_{n}$. The algorithm performs the following operations.
\begin{enumerate}
\item If $t$ is not in parallel form, then, using permutation
reductions, $t$ is reduced to a $t'$ which is in parallel form and
$\nor{N}(t')$ is recursively executed.
  
\item  \label{dot:simply} If $t$ is a simply typed $\lambda$-term, it is normalized and returned. If $t=u\parallel_{a} v$ is not a redex, then let $\nor{N}(u)=u'$ and  $\nor{N}(v)=v'$. If $u'\parallel_{a} v'$ is normal, it is returned. Otherwise, $\nor{N}(u'\parallel_{a} v')$ is recursively executed.

\item If $t$ is a redex, 
we select the \emph{smallest}
subterm $w$ of $t$ having maximal
$\mathcal{A}$-communication-complexity $r$. A sequence of terms is produced 
\[w\succ w_{1}\succ w_{2}\succ \ldots \succ w_{n}\]
such that $w_{n}$ has $\mathcal{A}$-communication-complexity strictly
smaller than $r$. We substitute $w_n$ for $w$ in $t$ obtaining $t'$
and recursively execute $\nor{N}(t')$.

\end{enumerate}  
\noindent We observe that in the step~\ref{dot:simply} of the algorithm $\nor{N}$, by
construction $u\parallel_{a} v$ is not a redex.  After $u$ and $v$ are
normalized respectively to $u'$ and $v'$, it can 
still be the case that $u'\parallel_{a} v'$ is not normal, because
some 
free variables of $u$ and $v$ may disappear
during the
normalization,
causing a new violation of the Subformula
Property that transforms $u'\parallel_{a} v'$ into a redex, even
though $u\parallel_{a} v$ was not.
 \end{definition}

\comment{
\begin{lemma}
Let $t: A$ be a term with free variables among
$$x_{1}^{A_{1}},\ldots, x_{n}^{A_{n}}$$
and $\mathcal{A}=\{A_{1}, \ldots, A_{n}, A\}$. Then $\succ$ is normalizing, that is, there are no infinite reduction chains
$$t\succ t_{0}\succ t_{1}\succ \ldots \succ t_{n}\succ\ldots$$
\end{lemma}
}

The first step of the normalization proof  
consists in showing that any term can be reduced to a parallel form.

\begin{proposition}\label{proposition-normpar}
Let $t: A$ be any term. Then $t\mapsto^{*} t'$, where $t'$ is a parallel form. 
\end{proposition}
\begin{proof}
Easy structural induction on $t$. 
See the Appendix.\end{proof}



We now prove that any term in parallel form can be normalized with the
help of the algorithm $\nor{N}$.
\begin{lemma}\label{lem:uebermensch} Let $t: A$ be a term in parallel
form which is not a simply typed $\lambda$-term and $\mathcal{A}$ a set of formulas
containing all proper subformulas of $A$ and closed under
subformulas. Assume that $r > 0$ is the maximum among the
$\mathcal{A}$-communication-complexities of the subterms of $t$. Assume
that the free variables $x_{1}^{A_{1}},\ldots, x_{n}^{A_{n}}$ of $t$
are such that for every ${i}$, either each strong subformula of
$A_{i}$ is in $ \mathcal{A}$, 
or each proper prime subformula of $A_{i}$ which has more than $r$ symbols  is in $ \mathcal{A}$.
  Suppose
moreover that no subterm $u\parallel_{a} v$ with
$\mathcal{A}$-communication-complexity $r$ contains a subterm of the
same $\mathcal{A}$-communication-complexity. Then there exists $t'$
such that $t \succ^* t' $ and the maximum among the
$\mathcal{A}$-communication-complexities of the subterms of $t'$ is
strictly smaller than $r$.
\end{lemma}
\begin{proof}
We prove the lemma by lexicographic induction on the pair
\[(\rho , k)\] 
where $k$ is the number of subterms of $t$ with
maximal $\mathcal{A}$-complexity $\rho$ among those with
$\mathcal{A}$-communication-complexity $r$. 



Let $u\parallel_a v$ be the \emph{smallest} subterm of $t$ having
$\mathcal{A}$-complexity $\rho$. Four cases can occur.

(a) $\rho =(r, d, l, o)$, with $d>0$. We first show that the term
$u\parallel_a v$ is a redex. Assume that the free variables of $u\parallel_a v$ are among
$x_{1}^{A_{1}},\ldots, x_{n}^{A_{n}}, a_{1}^{B_{1}\rightarrow
C_{1}},\ldots, a_{m}^{B_{m}\rightarrow C_{m}}$ and that the
communication kind of $a$ is $C, D$.

Suppose by contradiction that all the prime factors of $C$ and
$D$ are proper subformulas of $A$ or strong subformulas of one among
$A_1, \ldots, A_n, B_1\rightarrow C_1, \ldots, B_m\rightarrow C_m$.
Given that $r>0$ there is a prime factor $P$ of $C$ or $D$ such that
$P$ has $r$ symbols and does not belong to $ \mathcal{A}$. The
possible cases are two: (i) $P$ is a proper subformula of a prime
proper subformula $A'_i $ of $A_i$ such that $ A'_i \notin
\mathcal{A}$; (ii) $P$, by Proposition~\ref{proposition-strongsubf},
is a proper subformula of a prime factor of $B_i$ or $C_i$. Suppose (i). Since $A'_{i}\notin \mathcal{A}$, by hypothesis the number of symbols of $A_i'$ is less than or equal to $r$, so
$P$ cannot be a proper subfomula of $A_i'$, which is a
contradiction. Suppose (ii). Then there is a prime factor of $B_i$ or
$C_i$ having a number of symbols greater than $r$. Since by hypothesis $a_i^{B_i\rightarrow C_i}$ is bound in $t$,  we conclude
that there is a subterm $w_1\parallel_{a_i} w_2$ of $t$ having
$\mathcal{A}$-complexity greater than $\rho$, which is absurd.


Now, since $d>0$, we may assume $u=w_1\parallel_b w_2$ 
(the case $v=w_1\parallel_b w_2$ is symmetric).
The term \[(w_1\parallel_b w_2)\parallel_a v\]
is then a redex of $t$ and by replacing it with
\begin{equation}\label{eq:1}
(w_1\parallel_a v)\parallel_b  (w_2\parallel_a v)
\end{equation} we obtain from $t$ a term $t'$ such that $t\succ t'$
according to Def.~\ref{defi-redstrategy}. We must verify that we can
apply to $t'$ the main induction hypothesis. Indeed, the reduction
$t\succ t'$ duplicates all the subterms of $v$, but all of their
$\mathcal{A}$-complexities are smaller than $r$, because $u\parallel_a
v$ by choice is the smallest subterm of $t$ having maximal
$\mathcal{A}$-complexity $\rho$. The two terms $w_1\parallel_a v$ and
$w_2\parallel_a v$ have smaller $\mathcal{A}$-complexity than $\rho$,
because they have numbers of occurrences of the symbol $\parallel$
strictly smaller than in $u\parallel_{a} v$. Moreover, the terms in
$t'$ with~\eqref{eq:1} as a subterm have, by hypothesis,
$\mathcal{A}$-communication-complexity smaller than $r$ and hence
$\mathcal{A}$-complexity smaller than $\rho$. Assuming that the
communication kind of $b$ is $F, G$, the prime factors of $F$ and $G$
that are not in $\mathcal{A}$ must have fewer symbols than the prime
factors of $C$ and $D$ that are not in $\mathcal{A}$, again because
$u\parallel_a v$ by choice is the smallest subterm of $t$ having
maximal $\mathcal{A}$-complexity $\rho$; hence, the
$\mathcal{A}$-complexity of~\eqref{eq:1} is smaller than
$\rho$. Therefore the number of subterms of $t'$ with
$\mathcal{A}$-complexity $\rho$ is strictly smaller than $k$. By
induction hypothesis, $t'\succ^{*} t''$, where $t''$ satisfies the
thesis.\\

(b) $\rho=(r, d, l, o)$, with $d=0$ and $l>0$. Since $d=0$, $u$ and
$v$ are simply typed $\lambda$-terms -- and thus strongly normalizable
\cite{Girard} -- so we may assume $u\mapsto^{*} u'\in \nf$ and
$v\mapsto^{*} v'\in \nf$ by a sequence of intuitionistic reduction
rules. By replacing in $t$ the subterm $u\parallel_a v$ with
$u'\parallel_a v'$, we obtain a term $t'$ such that $t\succ t'$
according to Definition \ref{defi-redstrategy}. Moreover, the terms in
$t'$ with $u'\parallel_a v'$ as a subterm have, by hypothesis,
$\mathcal{A}$-communication-complexity smaller than $r$ and hence
$\mathcal{A}$-complexity is smaller than $\rho$. By induction
hypothesis, $t'\succ^{*} t''$, where $t''$ satisfies the thesis.\\

(c) $\rho=(r, d, l, o)$, with $d=l=0$. Since $d=0$, $u$ and $v$ are
simply typed $\lambda$-terms. Since $l=0$, $u$ and $v$ are in normal
form and thus satisfy conditions 1. and 2. of
Proposition~\ref{proposition-boundhyp}.  We need to check that
$u\parallel_a v$ is a redex, in particular that the communication
complexity of $a$ is greater than $0$. Assume that the free variables
of $u\parallel_a v$ are among $x_{1}^{A_{1}},\ldots, x_{n}^{A_{n}},
a_{1}^{B_{1}\rightarrow C_{1}},\ldots, a_{m}^{B_{m}\rightarrow C_{m}}$
and that the communication kind of $a$ is $C, D$. As we argued above,
we obtain that not all the prime factors of $C$ and $D$ are proper
subformulas of $A$ or strong subformulas of one among $A_1, \ldots,
A_n, B_1\rightarrow C_1, \ldots, B_m\rightarrow C_m$. By Definition
\ref{definition-comcomplexity}, $u\parallel_a v$ is a redex.

We now prove that every occurrence of $a$ in $u$ and $v$ is of the
form $a\, \xi$ for some term or projection $\xi$. First of all, $a$
occurs with arrow type both in $u$ and $v$. Moreover, $u: A$ and $v:
A$, since $t: A$ and $t$ is a parallel form; hence, the types
$C\rightarrow D$ and $D\rightarrow C$ cannot be subformulas of $A$,
otherwise $r=0$, and cannot be proper subformulas of one among $A_{1},
\ldots, A_{n}, B_{1}\rightarrow C_{1}, \ldots, B_{n}\rightarrow
C_{n}$, otherwise the prime factors of $C, D$ would be strong
subformulas of one among $A_1, \ldots, A_n, B_1\rightarrow C_1,
\ldots, B_m\rightarrow C_m$. Thus by Prop.~\ref{proposition-app} we
are done. Two cases can occur.

\begin{itemize}
\item $a$ does not occur in $u$ or $v$: to fix ideas, let us say it
does not occur in $u$. By performing a cross reduction, we replace in
$t$ the term $u\parallel_{a} v$ with $u$ and obtain a term $t'$ such
that $t\succ t'$ according to Def.~\ref{defi-redstrategy}. After the replacement, the
number of subterms having maximal $\mathcal{A}$-complexity $\rho$ in $t'$ is
strictly smaller than the number of such subterms in $t$. By induction
hypothesis, $t'\succ^{*} t''$, where $t''$ satisfies the thesis.

\item $a$ occurs in $u$ and in $v$. Let $u=\mathcal{C}[a\, w_1\,
\sigma]$ and $v=\mathcal{D}[a\, w_2\,\tau]$ where the displayed
occurrences of $a$ are the rightmost in $u$ and $v$ and $\sigma ,
\tau$ are the stacks of \emph{all} terms or projections $a$ is applied
to. By applying a cross reduction to  
$\mathcal{C}[a\, w_{1}\,\sigma]\parallel_{a} \mathcal{D}[a\, w_{2}\,\tau]$
we obtain the term $(\ast)$
$$ \small (
\mathcal{D}[w_1^{b\langle \sq{z}\rangle / \sq{y}}\,\tau]
\parallel_{a} 
\mathcal{C}[a\, w_1])\, \parallel_{b}\, (\mathcal{C}[w_2^{b\langle \sq{y}\rangle / \sq{z}}\,\sigma]\parallel_{a} \mathcal{D}[a\, w_2])$$
By hypothesis, $\sq{y}$ is the sequence of the free variables of $w_1$
which are bound in $\mathcal{C}[a\, w_1\,\sigma]$ and $\sq{z}$ is the
sequence of the free variables of $w_2$ which are bound in
$\mathcal{D}[a\, w_2\,\tau]$ and $a$ does not occur neither in $w_1$
nor in $w_2$. Since $u, v$ satisfy conditions 1. and 2. of 
Proposition~\ref{proposition-boundhyp}
the types $Y_{1}, \ldots, Y_{i}$ and $Z_{1}, \ldots, Z_{j}$ of
respectively the variables $\sq{y}$ and $\sq{z}$ are proper
subformulas of $A$ or strong subformulas of the formulas $A_{1},
\ldots, A_{n}, B_{1}\rightarrow C_{1}, \ldots, B_{m}\rightarrow
C_{m}$. Hence, the types among $Y_{1}, \ldots, Y_{i}, Z_{1},
\ldots, Z_{j}$ which are not in $\mathcal{A}$ are strictly smaller
than all the prime factors of the formulas $B_1, C_1, \ldots, B_m,
C_m$. Since the communication kind of $b$ is $Y_{1}\land\ldots \land
Y_{i}, Z_{1}\land \ldots\land Z_{j}$, by Definition
\ref{defi-acomplexity} either the $\mathcal{A}$-complexity of the term
$(\ast)$ above is strictly smaller than the $\mathcal{A}$-complexity
$\rho$ of $u\parallel_{a} v$, or the communication kind of $b$ is
$\top$. In the latter case we apply a cross reduction
$u_1 \parallel_{b} v_1 \mapsto u_1$ or $u_1 \parallel_{b} v_1 \mapsto
v_1$ and obtain a term with $\mathcal{A}$-complexity strictly smaller
than $\rho$.

In the former case, let $w_{1}', w_{2}'$ be simply typed $\lambda$-terms such that
$$w_1^{b\langle \sq{z}\rangle / \sq{y}}\,\tau\mapsto^{*} w_{1}'\in\nf \; \mbox{and} \;
w_2^{b\langle \sq{y}\rangle / \sq{z}}\,\sigma\mapsto^{*} w_{2}'\in \nf$$
By hypothesis, $a$ does not occur in $w_{1}, w_{2}, \sigma, \tau$ and
thus neither in $w_{1}'$ nor in $w_{2}'$. Moreover, by the assumptions
on $\sigma$ and $\tau$ and since $\mathcal{C}[a\, w_1\,\sigma]$ and
$\mathcal{D}[a\, w_2\,\tau]$ are normal simply typed $\lambda$-terms,
$\mathcal{C}[w_{2}']$ and $\mathcal{D}[ w_1']$ are normal too and
contain respectively one fewer occurrence of $a$ than the former
terms. Hence, the $\mathcal{A}$-complexity of the terms
$$ \mathcal{D}[w_1'] \parallel_{a} \mathcal{C}[a\, w_1] \quad \mbox{and} \quad
\mathcal{C}[w_2']\parallel_{a} \mathcal{D}[a\, w_2]$$
is strictly smaller than the $\mathcal{A}$-complexity $\rho$ of $u\parallel_{a} v$. Let now $t'$ be the term obtained from $t$ by replacing the term 
$\mathcal{C}[a\, w_{1}\,\sigma]\parallel_{a} \mathcal{D}[a\, w_{2}\,\tau]$
with
\begin{equation}\label{eq:very_big}
 (\mathcal{D}[w_1'] \parallel_{a} \mathcal{C}[a\,
w_1])\, \parallel_{b}\,(\mathcal{C}[w_2']\parallel_{a} \mathcal{D}[a\,
w_2])
\end{equation}
By construction $t\succ t'$.  Moreover, the terms in $t'$
with~\eqref{eq:very_big} as a subterm have, by hypothesis,
$\mathcal{A}$-communication-complexity smaller than $r$ and hence
$\mathcal{A}$-complexity is smaller than $\rho$. Hence, we can apply
the main induction hypothesis to $t'$ and obtain by induction
hypothesis, $t'\succ^{*} t''$, where $t''$ satisfies the thesis.
\end{itemize}

(d) $\rho =(r, d, l, o)$, with $d=l=o=0$. Since $o=0$, $u\parallel_a
v$ is a redex. To fix ideas, let us say $a$ does not occur in $u$.  By
performing a cross reduction, we replace $u\parallel_a v$ with $u$ so
that $ u\parallel_a v \succ u$ according to
Def.~\ref{defi-redstrategy}. Hence, by induction
hypothesis, $t'\succ^{*} t''$, where $t''$ satisfies the thesis.
\end{proof}

\begin{proposition}\label{proposition-normcomp} Let $t: A$ be any term
in parallel form. Then $t\mapsto^{*} t'$, where $t'$ is a
parallel normal form.
\end{proposition}
\begin{proof} Assume that the free variables of $t$ are
$x_{1}^{A_{1}},\ldots, x_{n}^{A_{n}}$ and let $\mathcal{A}$ be the set
of the proper subformulas of $A$ and the strong subformulas of the
formulas $A_{1}, \ldots, A_{n}$.  We prove the theorem by
lexicographic induction on the quadruple
\[(|\mathcal{A}|, r, k, s)\] where $|\mathcal{A}|$ is the cardinality
of $\mathcal{A}$, $r$ is the maximal
$\mathcal{A}$-communication-complexity of the subterms of $t$, $k$ is
the number of subterms of $t$ having maximal
$\mathcal{A}$-communication-complexity $r$ and $s$ is the size of
$t$. If $t$ is a simply typed $\lambda$-term, it has a normal form
\cite{Girard} and we are done; so we assume $t$ is not. There are two
main cases. \\
\medskip

\noindent \textit{First case}: $t$ \emph{is not a redex}.  Let
$t=u\parallel_{a} v$ and let $B, C$ be the communication kind of $a$.
Then, the communication complexity of $a$ is $0$ and by Def.~\ref{definition-comcomplexity} every prime factor of $B$ or $C$
belongs to $\mathcal{A}$.  Let $\mathcal{A}'$ be the set of the proper
subformulas of $A$ and the strong subformulas of the formulas $A_{1},
\ldots, A_{n}, B\rightarrow C$; let $\mathcal{A}''$ be the set of the
proper subformulas of $A$ and the strong subformulas of the formulas
$A_{1}, \ldots, A_{n}, C\rightarrow B$. By Prop.~\ref{proposition-strongsubf}, every strong subformula of $B\rightarrow
C$ or $C\rightarrow B$ is a proper subformula of a prime factor
of $B$ or $C$, and this prime factor is in
$\mathcal{A}$. Hence, $\mathcal{A}'\subseteq \mathcal{A}$ and
$\mathcal{A}''\subseteq \mathcal{A}$.

 If $\mathcal{A}'= \mathcal{A}$, then the maximal
$\mathcal{A}'$-communication-complexity of the terms of $u$ is less than or equal to $r$
and the number of terms having maximal $\mathcal{A}'$-communication-complexity is
less than or equal to $k$; since the size of $u$ is strictly smaller than
that of $t$, by induction hypothesis $u\mapsto^{*} u'$, where $u'$ is
a normal parallel form.

If $\mathcal{A}'\subset\mathcal{A}$, again by induction hypothesis
$u\mapsto^{*} u'$, where $u'$ is a normal parallel form.

The very same argument on $\mathcal{A}''$ shows that $v\mapsto^{*}
v'$, where $v'$ is a normal parallel form.

Let now $t'=u'\parallel_{a} v'$, so that $t\mapsto^{*} t'$. If $t'$ is
normal, we are done. If $t'$ is not normal, since $u'$ and $v'$ are
normal, the only possible redex remaining in $t'$ is the whole term
itself, i.e., $u'\parallel_{a} v'$: that happens only if the free
variables of $t'$ are fewer than those of $t$; w.l.o.g., assume they
are $x_{1}^{A_{1}}, \ldots, x_{i}^{A_{i}}$, with $i< n$.  Let
$\mathcal{B}$ be the set of the proper subformulas of $A$ and the
strong subformulas of the formulas $A_{1}, \ldots, A_{i}$. Since $t'$
is a redex, the communication complexity of $a$ is greater than $0$;
by Definition \ref{definition-comcomplexity}, a prime factor of $B$ or
$C$ is not in $\mathcal{B}$, so we have $\mathcal{B}\subset
\mathcal{A}$. By induction hypothesis, $t'\mapsto^{*}t''$, where $t''$
is a parallel normal form.
  
\bigskip

\noindent \textit{Second case}: $t$ \emph{is a redex}.  Let
$u\parallel_a v$ be the \emph{smallest} subterm of $t$ having
$\mathcal{A}$-communication-complexity $r$. The free variables of
$u\parallel_a v$ satisfy the hypotheses of Lemma~\ref{lem:uebermensch}
either because have type $A_i$ and $\mathcal{A}$ contains all the
strong subformulas of $A_i$, or because the prime proper subformulas
of their type have at most $r$ symbols, by
maximality of $r$. By Lemma~\ref{lem:uebermensch} $ u\parallel_a v
\succ^* w $ where the maximal among the
$\mathcal{A}$-communication-complexity of the subterms of $w$ is
strictly smaller than $r$. Let $t'$ be the term obtained replacing $w$
for $u\parallel_a v$ in $t$. We can now apply the induction hypothesis
and obtain $t' \mapsto^*t''$ with $t''$ in parallel normal form.
\end{proof}
The normalization for $\lamg$, and thus for $\NGod$, easily follows.

\begin{theorem}\label{theorem-normalization} Suppose that  $ t: A$ is a proof term  of $\LC$. Then $t\mapsto^{*} t': A$, where $t'$ is a normal parallel form.
\end{theorem}

\section{Computing with $\lamg$}
We illustrate the expressive power of $\lamg$ by a few examples.
All the examples employ the normalization algorithm in Definition \ref{defi-mastredstrategy};
to limit its non-determinism, when we have to reduce $u\parallel_{a}v$ because $a$ does not occur neither in $u$ nor in $v$, we always use the reduction $u\parallel_{a}v \mapsto u$. 

Henceforth we use the types $\mathbb{N}$ for natural
numbers, $\mathsf{Bool}$ for the Boolean values and $\mathsf{String}$ for strings.

We start by showing that $\lamg$ is more expressive than 
simply typed $\lambda$-calculus.
\begin{example}[\textbf{Parallel or}]\label{ex:parallel_or}
Berry's sequentiality theorem (see \cite{Girard}) implies that there
is no $\lambda$-term $\mathsf{O}: \mathsf{Bool} \IMPL \mathsf{Bool}
\IMPL\mathsf{Bool} $ such that $\mathsf{O}\mathsf{F}\mathsf{F} \mapsto \mathsf{F}$, $\mathsf{O}u\mathsf{T} \mapsto \mathsf{T}$,
$\mathsf{O}\mathsf{T}u \mapsto \mathsf{T}$,   
where $u$ is an arbitrary normal term, and thus possibly a variable.
 $\mathsf{O}$ can instead be defined in
Boudol's parallel $\lambda$-calculus~\cite{Boudol89}.

The $\lamg$ term for such parallel or is
(as usual the term ``$\mathsf{if} \, u \, \mathsf{then} \, s \,
\mathsf{else} \, t$'' reduces to $s$ if
$u = \mathsf{T}$, and to $t$ if $u = \mathsf{F}$):

\begin{footnotesize}
  \[ \mathsf{O} := \lambda x^{\mathsf{Bool}} \, \lambda
    y^{\mathsf{Bool}} \, (\mathsf{if} \, x \, \mathsf{then} \,
    (\lambda z \, \lambda k \, z ) \, \mathsf{else} \, (\lambda z \,
    \lambda k \, k ) )  \mathsf{T} (ax)  \]
  \[{ } \hspace{58pt}  \parallel_{a} (\mathsf{if} \, y \, \mathsf{then} \, (\lambda
    z \, \lambda k \, z ) \, \mathsf{else} \, (\lambda z \, \lambda k
    \, k ) ) \mathsf{T} (ay)  \]
\end{footnotesize}where the communication kind of $a$ is $
\mathsf{Bool} 
, \mathsf{Bool}$. Now \begin{small}\begin{align*}
                          \mathsf{O}\, u \,\mathsf{T}     \mapsto^{*} & (\mathsf{if} \, u \, \mathsf{then} \, (\lambda z \, \lambda k \,
    z ) \, \mathsf{else} \, (\lambda z \, \lambda k \, k )
    )\mathsf{T}(au)
\\
&    \parallel_{a} (\mathsf{if} \, \mathsf{T} \, \mathsf{then} \,
    (\lambda z \, \lambda k \, z ) \, \mathsf{else} \, (\lambda z \,
    \lambda k \, k ) )\mathsf{T}(a \mathsf{T})
  \\
\mapsto^{*} &  (\mathsf{if} \, u \, \mathsf{then} \, (\lambda z
    \, \lambda k \, z ) \, \mathsf{else} \, (\lambda z \, \lambda k \,
    k ) )\mathsf{T}(au ) \parallel_{a} \mathsf{T} \,  \mapsto \,
    \mathsf{T}
  \end{align*}
\end{small}And symmetrically $\mathsf{O}\,  \mathsf{T}\, u \, \mapsto^{*} \, \mathsf{T} $.
On the other hand\begin{small}
  \begin{align*}
\mathsf{O} \, \mathsf{F} \,\mathsf{F}  \mapsto^{*} & \quad 
    (\lambda z \, \lambda k \, k ) \mathsf{T}(a  \mathsf{F})
    \parallel_{a} (\lambda z \, \lambda k \, k ) \mathsf{T}( a
    \mathsf{F})
\\
\mapsto^{*} & \quad ( a \mathsf{F})
    \parallel_{a} ( a \mathsf{F} )
\\
\mapsto^{*} & \quad
    (\mathsf{F}
    \parallel_{a} ( a  \mathsf{F})
    )\parallel _{b} (\mathsf{F} \parallel_{a}   ( a \mathsf{F})  ) \quad \mapsto^{*} \quad \mathsf{F}
  \end{align*}\end{small}

\end{example}

\begin{example}[\textbf{Data passing}]\label{ex:pi_calc}
As in the previous example, if the messages sent during a cross
reduction are closed terms, for example data, the outcome is a simple unidirectional
message passing. Indeed, the newly introduced communication is void
and is always removed: 
$$\mathcal{C}[a\, u]\parallel_{a} \mathcal{D}[a \, v] \\ \mapsto$$ 
$$(\mathcal{D}[u] \parallel_{a}
\mathcal{C}[a\, u]) \parallel_{b} (\mathcal{C}[v]\parallel_{a}
\mathcal{D}[a\, v]) \, \mapsto \, \mathcal{D}[u] \parallel_{a}
\mathcal{C}[a\, u]$$ 
If we want a process $s$ to transmit a message $m:B$ to a process $t$ 
without $ t$ passing anything back, we can use the following term ($a$ has communication kind  $(B\IMPL
F)\IMPL F, F\IMPL F$):
\begin{small}
  \[(a \lambda z^{A\rightarrow F }
 \, zm)s \parallel_{a} (a \lambda y^{F} \, y)(\lambda x^{B} \, t)
 \; \mapsto \]\[
    ( (\lambda z \, zm)(\lambda x\, t)
    \parallel_{a} (a \lambda z \, zm)s )
    \parallel_{e}
    ((\lambda y\, y)s  \parallel_{a} (a \lambda
    y \, y)(\lambda x\, t))\]
  \[ \mapsto^{*} \quad (\lambda z \, zm)(\lambda x\,
    t)\parallel_{e} (\lambda y\, y)s \quad \mapsto^{*} \quad t[m/x]\parallel_{e} s\]
\end{small} \noindent This reduction resembles indeed the unidirectional communication $\overline{a}\langle m\rangle.P \mid a(x).Q \mapsto P \mid Q[m/x]$ in the $\pi$-calculus \cite{Milner, sangiorgiwalker2003}, assuming $a$ does not occur in $P$ and $Q$.
\end{example}

In the following example, similar to that in \cite{CP2010}, we
simulates the communication needed to conclude an online sale.
\begin{example}[\textbf{Buyer and vendor}]\label{ex:sale} 
We model the following transaction: a buyer tells a vendor a product
name $\mathsf{prod}:\mathsf{String}$, the vendor computes the value
$\mathsf{price}:\mathbb{N}$ of $\mathsf{prod}$ and sends it to the
buyer, the buyer sends back the credit card number
$\mathsf{card}: \mathsf{String}$ which is used to pay.

We introduce the following  functions:
$\mathsf{cost} : \mathsf{String}\rightarrow \mathbb{N} $ with input a product name $\mathsf{prod}$ and output its cost $\mathsf{price}$; $\mathsf{pay\_for}: \mathbb{N}\rightarrow
\mathsf{String} $ with input a $\mathsf{price}$ and output a credit card
number $\mathsf{card}$; $\mathsf{use}: \mathsf{String}\rightarrow \mathbb{N}$ that obtains money
using as input a credit card number $\mathsf{card}: {\mathsf{String}} $.  The buyer and the
vendor are the contexts $\mathcal{B}$ and $\mathcal{V}$ of type
$\mathsf{Bool}$.
Notice that the terms representing buyer and vendor exchange their
position at each cross reduction. For $a$ of kind
$\mathsf{String} , \mathbb{N}$, the program is: \begin{small}\begin{align*} & \mathcal{B}[a( \mathsf{pay\_for}(a(\mathsf{prod})))] \parallel_{a} \mathcal{V}[\mathsf{use}( a(
    \mathsf{cost}(a\, 0)))]
\\
&\mapsto^{*} \mathcal{V}[\mathsf{use}( a( \mathsf{cost}(
      \mathsf{prod} ))) ] \parallel_{a}
      \mathcal{B}[a( \mathsf{pay\_for}(a(\mathsf{prod})))]
\\ &\mapsto \mathcal{V}[\mathsf{use}( a( \mathsf{price}))
]\parallel_{a} \mathcal{B}[a( \mathsf{pay\_for}(a(\mathsf{prod})))] \\
&\mapsto^{*}  \mathcal{B}[a( \mathsf{pay\_for}(\mathsf{price}))
] \parallel_{a} \mathcal{V}[\mathsf{use}( a( \mathsf{price}))] \\
&\mapsto  \; \mathcal{B}[a (\mathsf{card})] \parallel_{a}
\mathcal{V}[\mathsf{use}( a( \mathsf{price}))] \, \mapsto^{*} \,
\mathcal{V}[\mathsf{use}( \mathsf{card})] \parallel_{a} \mathcal{B}[a(
\mathsf{card})]
  \end{align*}\end{small}Finally $ \mapsto
\mathcal{V}[\mathsf{use}( \mathsf{card})]$: the buyer has
performed its duty and the vendor uses the card number to obtain the
due payment. 
\end{example}

We show that although more complicated than sending data, 
sending open processes can enhance efficiency.
\begin{example}[\textbf{Efficiency via cross
reductions}]\label{ex:code_mobility} 
Given three processes $M \parallel_{d} (
P \parallel_{a} Q)$. Assume that $Q$ wants to send a process to $P$, but one of
the process' parameters is not available because $M$ first needs many
time-consuming steps to produce it and only afterwards can send it to
$Q$. Cross reductions make it possible to fully exploit parallelism
and improve the program efficiency: $Q$ does not need to wait that
much and can send the process directly to $P$, which can begin to
partially evaluate it with no further delay. After having computed
the data, $M$ sends it to $Q$ which in turn forwards it to $P$.

For a concrete example, assume that
\begin{small}
  \begin{align*} M& \quad \mapsto^{*}\; d \, (\lambda k^{\mathbb{N}
\IMPL \mathbb{N} \IMPL \mathbb{N}} \, k \, 7 \, 0 ) \\ P & \quad =
\quad d \, 0 \, (\lambda j^{\mathbb{N}} \, \lambda x^{\mathbb{N}} \,
(ax)5s) \\ Q & \quad = \quad d \, 0 \, (\lambda y^{\mathbb{N}} \,
\lambda l^{\mathbb{N}} \, a(\lambda z^{\mathbb{N}} \, \lambda i
^{\sigma} \, h \langle g(z),y \rangle ))
  \end{align*}
\end{small}where \mbox{$h : \mathbb{N} \ET \mathbb{N} \IMPL
\mathbb{N}$}, \mbox{$g : \mathbb{N} \IMPL \mathbb{N}$}, the
communication kind of $d$ is $ (\mathbb{N} \IMPL \mathbb{N} \IMPL
\mathbb{N}) \IMPL \mathbb{N} , \mathbb{N}$, and the communication kind
of $a$ is $\mathbb{N} , \mathbb{N} \IMPL \sigma \IMPL \mathbb{N}$ with
$\sigma$ arbitrary type of high complexity.  Here $Q$ wants to send
$\lambda z^{\mathbb{N}} \, \lambda i ^{\sigma} \, h \langle g(z),y \rangle $ to $P$,
but the value $7$ of the parameter $y$ is computed and transmitted to $Q$ by $M$
only later. On the other hand, $P$ waits for the process from $Q$ in order to instantiate $z$
with $5$ and compute $h \langle g(5), 7 \rangle$.

Without a special mechanism for sending open terms, $P$ must wait for $M$ to
normalize. Afterwards $M$ passes $(\lambda k \, k \, 7 \, 0)$
through $d$ to $P$ and $Q$ with the following computation:
\begin{footnotesize}
  \begin{align*}
    M \parallel_{d} ( P \parallel_{a} Q) \mapsto^{*} & (\lambda k \, k
\, 7 \, 0) (\lambda j \, \lambda x\, (ax)5s)
                  \parallel_{a} \\ &
                  (\lambda k \, k \,  7 \, 0) (\lambda y \,
                  \lambda l \, a(\lambda z \, \lambda i \, h \langle g(z),y \rangle )) \mapsto^{*} 
\\
(a\,0)5s \parallel  _{a}  a(\lambda z \, \lambda i \, h \langle g & (z),7 \rangle )
\mapsto^{*} 
(\lambda z \, \lambda i \, h \langle g(z),7 \rangle )5s \mapsto^{*}  h \langle g(5),7 \rangle 
  \end{align*}
\end{footnotesize}
Our normalization algorithm allows instead  $Q$ to directly send $\lambda z^{\mathbb{N}} \, \lambda i
^{\sigma} \, h \langle g(z),y \rangle $ to $P$ by executing first a cross reduction:
\begin{footnotesize}\begin{align*}
         & M 
    \parallel_{d} \Big( 
d \, 0 (\lambda j \,  \lambda
 x \, (ax)5s)
\parallel_{a}
d \, 0 (\lambda y\, \lambda l \,  a(\lambda
z \, \lambda i \, h \langle g(z),y\rangle)) 
\Big) 
 \\
 \mapsto &
      M 
      \parallel_{d}
\Big(
 (d \, 0 (\lambda y\, \lambda l \,
by
) ||_{a} P )
\parallel_{b}
\big(d \, 0 (\lambda j \, \lambda
 x \, (
\lambda
z \, \lambda i  \, h \langle g(z),bx \rangle
)5s) \parallel_{a} Q\big)\Big) \\
    \mapsto&^{*}
      M 
      \parallel_{d}
\Big(
d \, 0 (\lambda y\, \lambda l \,
by
)
\parallel_{b}
d \, 0 (\lambda j \, \lambda
 x \, (
\lambda
z \, \lambda i  \, h \langle g(z),bx \rangle
)5s) \Big) 
    \end{align*}
  \end{footnotesize}  where the communication $b$ (of kind $\mathbb{N} ,
\mathbb{N}$) redirects the data $x$ and $y$. Then $P$ instantiates $z$
with $5$ and can compute for example $g(5)=9$ without having to
evaluate $h \langle g(5),7 \rangle $ all at once. When $M$ terminates
the computation, sends $7$ to the new location of the partially
evaluated processes $P$ and $Q$ via $\parallel_{b}$:\begin{small}
  \begin{align*}
\mapsto^{*} &
 M 
      \parallel_{d}
\Big(
d \, 0 (\lambda y\, 
\lambda l \,  by
)
\parallel_{b}
d \, 0 (\lambda j \,\lambda
 x \, h \langle g(5),bx \rangle 
)
 \Big)
\\
\mapsto^{*}&
\Big(
M
\parallel_{d}
d \, 0 (\lambda y\, 
\lambda l \,  by
)\Big) 
\parallel_{b}
\Big( M 
      \parallel_{d}
d \, 0 (\lambda j \,\lambda
 x \, h \langle 9,bx \rangle 
)
\Big)
\\
\mapsto^{*} &
(\lambda k \, k \, 7 \, 0) (\lambda y\, 
\lambda l \, by
)
\parallel_{b}
(\lambda k \, k \, 7 \, 0) (\lambda
 j \, \lambda
 x \,  h \langle 9,bx \rangle 
)
\\
\mapsto^{*}&
\quad b7
\parallel_{b}
h \langle 9,b\, 0 \rangle 
\quad \mapsto^{*} \quad 
h \langle 9,7 \rangle 
\end{align*}
\end{small}
\end{example}

\noindent
{\bf Final Remark}
The Curry--Howard isomorphism for  
$\lamg$ interprets G\"odel logic in terms of  
communication  between parallel processes. In addition to revealing
this connection, our results pave the way towards a more general
computational interpretation of the intermediate logics formalized by
hypersequent calculi. These logics are characterized by disjunctive
axioms of a suitable form~\cite{CGT08} -- containing all the
disjunctive tautologies of~\cite{DanosKrivine} -- and likely
correspond to other communication mechanisms between parallel
processes.

%
%




\begin{thebibliography}{50}







\bibitem
  {Aschieri2016} F. Aschieri. On Natural
Deduction for Herbrand Constructive Logics I: Curry--Howard
Correspondence for Dummett's Logic LC. \newblock \textit{Log.\ Methods Comput.\ Sci.}, vol. 12(3) 
n.~13, pp. 1--31,  2016.


\bibitem
{AschieriZH} F. Aschieri, M. Zorzi. On
Natural Deduction in Classical First-Order Logic: Curry--Howard
Correspondence, Strong Normalization and Herbrand's Theorem.
\emph{Theoret.\ Comput.\ Sci.}, 625: 125--146, 2016.



\bibitem
  {Avron91} A. Avron. Hypersequents, logical
consequence and intermediate logics for
concurrency. \emph{Ann. Math. Artif. Intell.}, 4: 225--248, 1991.


\bibitem
  {Avron96} A. Avron.
\newblock The method of hypersequents in the proof theory of propositional non-classical logic.
\newblock In 
\textit{Logic: From Foundations to Applications}.
\newblock Oxford University Press, pp. 1--32, 1996.

\bibitem
  {HyperAgata} M. Baaz, A. Ciabattoni, C. Ferm\"uller. A Natural Deduction System for Intuitionistic Fuzzy Logic. In 
\emph{Lectures on Soft Computing and Fuzzy Logic}, pp. 1--18, Physica-Verlag, 2000.
 

\bibitem
  {BP2015} A. Beckmann and
N. Preining. Hyper natural deduction. In \emph{LICS
  2015}, pp. 547--558, 2015. 



\bibitem{Boudol89} G. Boudol. Towards a lambda-calculus for concurrent and communicating systems.
\emph{TAPSOFT '98}, pp. 149-161, vol. 1, 1989.



\bibitem{CP2010}
L. Caires and F. Pfenning. Session types as intuitionistic 
linear propositions. In \textit{CONCUR 2010} 
pages 222–-236. LNCS 6269, 2010.




\bibitem{CGT08} A. Ciabattoni, N. Galatos and
K. Terui.  \newblock From axioms to analytic rules in nonclassical
logics.  \newblock In \textit{LICS 2008}, pp. 229--240, 2008.





\bibitem{CG2016}
A. Ciabattoni and F. A. Genco.
\newblock Embedding formalisms: hypersequents and two-level systems of rules.
\newblock In {\em AIML 2016}, pp. 197--216, 2016.


\bibitem{DanosKrivine} V. Danos, J.-L. Krivine. Disjunctive Tautologies as Synchronisation Schemes. In \emph{CSL 2000}, 1862: 292--301, 2000.



\bibitem{dummett} M. Dummett. A propositional calculus
with denumerable matrix. {\em J.\ Symbolic Logic}, 24: 97--106, 1959.

\bibitem{Javascript} D. Flanagan. JavaScript: the Definitive Guide, O'Reilly Media, 2011.


\bibitem{codemobil} A. Fuggetta, G.P. Picco and G. Vigna. Understanding Code Mobility. \emph{IEEE Transactions on Software Engineering}, 24: 342--361, 1998.



\bibitem{Girard}
J.-Y. Girard and Y. Lafont and P. Taylor,
\emph{Proofs and Types}.
Cambridge University Press,
1989.


\bibitem{goedel} K. G\"odel: Zum intuitionistischen Aussagenkalk\"ul.
{\em Anzeiger der Kaiserlichen Akademie der Wissenschaften, Mathematisch-Naturwissenschaftliche Classe. Wien}. 
69: 65--66, 1932.





\bibitem
  {Griffin} T. Griffin. A Formulae-as-Type Notion of Control. In \emph{POPL 1990}, 1990.


\bibitem
  {deGrooteex} P. de Groote, {A Simple Calculus of Exception Handling}, \emph{Proceedings of TLCA 1995}, LNCS, vol. 902 pp. 201--215, 1995.

\bibitem
  {Hirai} Y. Hirai. A lambda calculus for
G\"{o}del--Dummett logic capturing waitfreedom. In \emph{FLOPS 2012},
pp. 151--165, 2012.

\bibitem
  {Howard} W. A. Howard. The formulae-as-types notion of construction. In \emph{
To H. B. Curry: Essays on Combinatory Logic, Lambda Calculus,
and Formalism}, Academic Press, pp. 479--491. 1980.

\bibitem
  {Krivine} J.-L. Krivine. Classical
Realizability. \emph{Interactive models of computation and program
behavior, Panoramas et synth\`eses}, pp. 197--229, 2009.  

\bibitem
  {Krivine1}
J.-L. Krivine. Lambda-calcul types et mod\`eles. \emph{Studies in Logic and Foundations of Mathematics}. Masson, pp. 1--176. 1990.
\bibitem{landin} P. J. Landin. The Mechanical Evaluation of Expressions. \emph{The Computer Journal}, 6(4), pp. 308--320, 1964.

\bibitem
  {L1982}
E.G.K. Lopez-Escobar. Implicational logics in natural deduction systems.
{\em J.\ Symbolic Logic}, 47(1): 184--186, 1982.


\bibitem
  {MGO}
G. Metcalfe and N. Olivetti and D. Gabbay.
Proof Theory for Fuzzy Logics. \emph{Springer Series in Applied Logic} vol. 36, 2008.

\bibitem
  {Milner}
R. Milner.
Functions as Processes. \emph{Mathematical Structures in Computer Science}, vol. 2, n. {2}, pp. {119--141},{1992}.



\bibitem{mostrousyoshida2015} D. Mostrous, N. Yoshida. Session typing and asynchronous subtyping for the
higher-order $\pi$-calculus. \emph{Inf.\ Comput.}, vol. 241, pp. 227--263, 2015.




\bibitem{MCHP} T. Murphy, K. Crary, R. Harper,
F. Pfenning.  A Symmetric Modal Lambda Calculus for Distributed
Computing. In \emph{LICS 2004}, pp. 286--295, 2004.




\bibitem{Negri:2014} S. Negri.  Proof analysis beyond
geometric theories: from rule systems to systems of rules.  {\em
J.\ Logic Comput.}, vol. 27, pp. 513-537, 2016.


\bibitem
  {Parigot} M. Parigot. Proofs of Strong
Normalization for Second-Order Classical Natural Deduction. \emph{J.\ Symbolic Logic}, 62(4): 1461--1479, 1997.



\bibitem
  {perez2015} J. A. P\'{e}rez. Higher-Order
Concurrency: Expressiveness and Decidability Results, PhD thesis,
University of Bologna, 2010.

\bibitem
  {Prawitz}
D. Prawitz. Ideas and Results in Proof Theory. In \emph{Proceedings of the Second Scandinavian Logic Symposium}, 1971.


\bibitem
  {sangiorgiwalker2003} 
D. Sangiorgi and D. Walker. The pi-calculus: a Theory of Mobile
Processes.  
2003.



\bibitem
  {Wadler} P. Wadler. Propositions as Types. \emph{Communications of the ACM}, 58(12): 75--84, 2015.




\end{thebibliography}
%

\appendix

\section{Appendix}




\noindent \textbf{Propostition~\ref{proposition-strongsubf}}~(Characterization of Strong Subformulas)\textbf{.}
Suppose $B$ is any {strong subformula} of $A$. Then:
\begin{itemize}
\item If $A=A_{1}\land \ldots \land A_{n}$, with $n>0$ and $A_{1}, \ldots, A_{n}$ are prime, then $B$ is a proper subformula of one among $A_{1}, \ldots, A_{n}$. 
\item If $A=C\rightarrow D$, then $B$ is a proper subformula of a prime factor of $C$ or $D$. 
\end{itemize}
\begin{proof}\mbox{}
\begin{itemize}
\item Suppose $A=A_{1}\land \ldots \land A_{n}$, with $n>0$ and $A_{1}, \ldots, A_{n}$ are prime. Any prime proper subformula of $A$ is a subformula of one among $A_{1}, \ldots, A_{n}$, so $B$ must be a proper subformula of one among $A_{1}, \ldots, A_{n}$.

\item Suppose  $A=C\rightarrow D$. Any prime proper subformula $\mathcal{X}$ of $A$ is first of all a subformula of $C$ or $D$. Assume now $C=C_{1}\land \ldots \land C_{n}$ and $D=D_{1}\land \ldots \land D_{m}$, with $C_{1}, \ldots, C_{n}, D_{1}, \ldots, D_{m}$ prime. Since  $\mathcal{X}$ is prime, it must be a subformula of one among $C_{1}, \ldots, C_{n}, D_{1}, \ldots, D_{m}$ and since $B$ is a proper subformula of $\mathcal{X}$, it must be a proper subformula of one among $C_{1}, \ldots, C_{n}, D_{1}, \ldots, D_{m}$, QED. 
\end{itemize}
\end{proof}


\noindent \textbf{Theorem}~\ref{subjectred}~(Subject Reduction)\textbf{.}
If $t : A$ and $t \mapsto u$, then
    $u : A$ and all the free variables of $u$ appear among those of
    $t$.
\begin{proof} 
\begin{enumerate}
\item  $(\lambda x\, u)t\mapsto u[t/x]$.

A term $(\lambda x\, u)t$ corresponds to a derivation of the form
\[\infer{(\lambda x\, u)t : B}{\infer[^{1}]{\lambda x\, u: A \IMPL
B}{\infer*{u:B}{\Gamma_{0} , [x:A]^{1}}} & & \deduce{t:
A}{\deduce{{\cal P}}{\Gamma_{1}}}}\] where $\Gamma_{0}$ and $\Gamma_{1}$ are the open
hypotheses. Hence, the term
$u[t/x]$ corresponds to
\[ \infer*{u:B} { \Gamma_{0} & \deduce{t: A}{\deduce{{\cal P}}{\Gamma_{1}}} }\]
where all occurrences of $[x:A]^{1}$ are replaced by the derivation
${\cal P}$ of $t: A$. It is easy to prove, by induction on the
length of the derivation from $\Gamma_{0} , x:A$ to $B$, that this is a
derivation of $B$. The two derivations have the same open hypotheses
$\Gamma_{0}$ and $\Gamma_{1}$, and hence the two terms $(\lambda x\, u)t$ and
$u[t/x]$ have the same free variables.  Indeed, it is easy to show
that if a term corresponds to a derivation, each free variable of such
term corresponds to an open hypothesis of the derivation.

\item $ \pair{u_0}{u_1}\pi_{i}\mapsto u_i$, for $i=0,1$.

A term $\pair{u_0}{u_1}\pi_{i}$ corresponds to a derivation of the
form
\[\infer{\pair{u_0}{u_1}\pi_{i} : A_{i} }{\infer{\pair{u_0}{u_1}:
A_{0} \ET A_{1}} { \infer*{u_{0}: A_{0}}{\Gamma_{0}} & &
\infer*{u_{1}: A_{1}}{\Gamma_{1}} }}
\] where $\Gamma_{0}$ and $\Gamma_{1}$ are the open hypotheses.
Hence, the term $u_i$ corresponds to
\[ \infer*{u_{i}:A_{i}} { \Gamma_{i}}\] Clearly, all open hypotheses
of the derivation of $u_{i}:A_{i}$ also occur in the derivation of
$\pair{u_0}{u_1}\pi_{i} : A_{i}$, hence the free variables of $u_{i}$
are a subset of those of $\pair{u_0}{u_1}\pi_{i}$.

\item $(\Ecrom{a}{u}{v}) w \mapsto \Ecrom{a}{uw}{vw}$, if $a$ does not
occur free in $w$.

A term $(\Ecrom{a}{u}{v}) w$ corresponds to a derivation of the
form
\[\infer{(\Ecrom{a}{u}{v}) w : D }{\infer[^{1}]{
\Ecrom{a}{u}{v}:C \IMPL D}{\infer*{u:C \IMPL D}{\Gamma_{0} , [a:A \IMPL B]^{1} } &&
\infer*{v:C \IMPL D}{\Gamma_{1} , [a:B \IMPL A]^{1} }} & &
\infer*{w: C}{\Gamma_{2}} }
\] where $\Gamma_{0} , \Gamma_{1}$ and $\Gamma_{2}$ are the open hypotheses.
Hence, the term $ \Ecrom{a}{uw}{vw}$ corresponds to
\[\infer[^{1}]
  {
    \Ecrom{a}{uw}{vw} : D
  }
  {
    \infer{uw:D}
    {
      \infer*{u:C \IMPL
        D}
      {
        \Gamma_{0} , [a:A \IMPL B]^{1}
      }
      &&
      \infer*{w: C}
      {\Gamma_{2}}
    }
    &&
    \infer{vw:D}
    {
      \infer*{v:C \IMPL D}{\Gamma_{1} , [a:B \IMPL A]^{1} }      
      &&
      \infer*{w: C}{\Gamma_{2}}
    }
  }
\] 
Notice that the formulas corresponding to the terms $a$ are discharged
above the appropriate premise of the $\COM$ rule application. All open hypotheses
of the second derivation also occur in the first derivation, hence the
free variables of  $(\Ecrom{a}{u}{v}) w$ are also free variables of $\Ecrom{a}{uw}{vw}$.

\item $w(\Ecrom{a}{u}{v}) \mapsto \Ecrom{a}{wu}{wv}$, if $a$ does not
occur free in $w$.

A term $w(\Ecrom{a}{u}{v})$ corresponds to a derivation of the
form
\[\infer{w(\Ecrom{a}{u}{v}): D }{\infer*{w: C \IMPL D}{\Gamma_{0}} && \infer[^{1}]{
\Ecrom{a}{u}{v}:C }{\infer*{u:C }{\Gamma_{1} , [a:A \IMPL B]^{1} } &&
\infer*{v:C}{\Gamma_{2} , [a:B \IMPL A]^{1} }}}
\]

 where $\Gamma_{0} , \Gamma_{1}$ and $\Gamma_{2}$ are the open hypotheses.
Hence, the term $ \Ecrom{a}{wu}{wv}$ corresponds to
\[\infer[^{1}]
  {
    \Ecrom{a}{wu}{wv} : D
  }
  {
    \infer{wu:D}
    {
      \infer*{w: C\IMPL D}{\Gamma_{0}}
&
      \infer*{u:C}{\deduce{[a:A \IMPL B]^{1}}{\Gamma_{1} ,}  }
    }
&
    \infer{wv:D}
    {
      \infer*{w: C  \IMPL
        D}{\Gamma_{0}}
&
      \infer*{v:C}{\deduce{[a:B \IMPL A]^{1}}{\Gamma_{2} ,}  }
    }
  }
\]
In which the formulas corresponding to the terms $a$ are discharged above the appropriate premise of the $\COM$ rule application. All open
hypotheses of the second derivation also occur in the first
derivation, and thus the free variables of
$\Ecrom{a}{wu}{wv}$ are also free variables of  $w(\Ecrom{a}{u}{v})$.

\item $\efq{P}{\Ecrom{a}{w_{1}}{w_{2}}} \mapsto
\Ecrom{a}{\efq{P}{w_{1}}}{\efq{P}{w_{2}}}$, where $P$ is atomic.

A term $\efq{P}{\Ecrom{a}{w_{1}}{w_{2}}} $ corresponds to a derivation of the
form
\[\vcenter{
\infer{\efq{P}{\Ecrom{a}{w_{1}}{w_{2}}}: P}{
\infer{\Ecrom{a}{w_{1}}{w_{2}} : \FAL}{\infer*{w_{1}: \FAL}{\Gamma
    _{1} , [a:A \IMPL B]^{1}} && \infer*{w_{2}: \FAL}{\Gamma _{2} , [a:B \IMPL A]^{1}}}}
}\]
 where $\Gamma_{1}$ and $\Gamma_{2}$ are the open hypotheses.
Hence, the term $\Ecrom{a}{\efq{P}{w_{1}}}{\efq{P}{w_{2}}}$ corresponds to
\[
\vcenter{\infer[^{1}]{\Ecrom{a}{\efq{P}{w_{1}}}{\efq{P}{w_{2}}} :
    P}{\infer{\efq{P}{w_{1}} : P}{\infer*{\FAL}{\Gamma _{1} , [a:A \IMPL B]^{1}}}
&&& \infer{\efq{P}{w_{1}}:P}{\infer*{\FAL}{\Gamma _{2} , [a:B \IMPL A]^{1}}}}}\] 
In which the formulas corresponding to the terms $a$ are discharged
above the appropriate premise of the $\COM$ rule application. All open
hypotheses of the second derivation also occur in the first
derivation, and thus
$\Ecrom{a}{\efq{P}{w_{1}}}{\efq{P}{w_{2}}}$ has the same free
variables as $\efq{P}{\Ecrom{a}{w_{1}}{w_{2}}} $.

\item 
$(\Ecrom{a}{u}{v})\pi_{i} \mapsto \Ecrom{a}{u\pi_{i}}{v\pi_{i}}$, for $i=0,1$.

A term $(\Ecrom{a}{u}{v})\pi_{i}$ corresponds to a derivation of the
form
\[\infer{ (\Ecrom{a}{u}{v})\pi_{i} : A_{i}}{
\infer[^{1}]{\Ecrom{a}{u}{v} : C_{0} \ET C_{1}}{\infer*{u:C_{0} \ET
C_{1}}{\Gamma _{1} , [a:A \IMPL B]^{1}} & \infer*{v:C_{0} \ET
C_{1}}{\Gamma _{2} , [a:B \IMPL A]^{1}} }}
\] where $\Gamma_{1}$ and $\Gamma_{2}$ are the open hypotheses.
Hence, the term $\Ecrom{a}{u\pi_{i}}{v\pi_{i}}$ corresponds to
\[ \infer{\Ecrom{a}{u\pi_{i}}{v\pi_{i}}: C_{i}}{
\infer{u\pi _{i}:C_{i}}{\infer*{u:C_{0} \ET C_{1}}{\Gamma _{1}
, [a:A \IMPL B]^{1}}} & \infer{v\pi
_{i}:C_{i}}{\infer*{v:C_{0} \ET C_{1}}{\Gamma _{2} , [a:B \IMPL
A]^{1}}}}
\] In which the formulas corresponding to the terms $a$ are discharged
above the appropriate premise of the $\COM$ rule application. All open
hypotheses of the second derivation also occur in the first
derivation, and thus the free variables of $ (\Ecrom{a}{u}{v})\pi_{i}
$ and of $ \Ecrom{a}{u\pi_{i}}{v\pi_{i}} $ are the same.

\item $\lambda x\,(\Ecrom{a}{u}{v}) \mapsto \Ecrom{a}{\lambda x\,u}{\lambda
x\, v} $.

A term $\lambda x\,(\Ecrom{a}{u}{v})$
corresponds to a derivation 
\[
\vcenter{\infer[^{2}]{\lambda x\,(\Ecrom{a}{u}{v}): C\IMPL
D}{\infer[^{1}]{\Ecrom{a}{u}{v}:D}{\infer*{u:D}{\Gamma _{1}
, [x:C]^{2}, [a:A \IMPL B]^{1} } &&& \infer*{v:D}{\Gamma _{2} , [x:C]^{2},  [a:B \IMPL
A]^{1}}}} }\]
where 
$\Gamma_{1}$ and $\Gamma_{2}$ 
are the open hypotheses.
Hence, the term
 $\Ecrom{a}{\lambda x\,u}{\lambda
x\, v} $
corresponds to
\[\infer[^{3}]{\Ecrom{a}{\lambda x\,u}{\lambda
x\, v} : C\IMPL D}{ \infer[^{1}] {\lambda x \,
u: C\IMPL D} {\infer*{u:D}{\Gamma _{1}
, [x:C]^{1}, [a:A \IMPL B]^{3} }} & \infer[^{2}]{\lambda x \,
v: C \IMPL D}{\infer*{v:D}{\Gamma _{2} , [x:C]^{2},  [a:B \IMPL
A]^{3}}}}
\]
 In which the formulas corresponding to the terms $a$ are discharged
above the appropriate premise of the $\COM$ rule application. All open
hypotheses of the second derivation also occur in the first
derivation, and thus $\lambda
x\,(\Ecrom{a}{u}{v})$ and $ \Ecrom{a}{\lambda x\,u}{\lambda
x\, v} $ have the same free variables.

\item $\langle u \parallel_{a} v,\, w\rangle \mapsto \langle u,
w\rangle \parallel_{a} \langle v, w\rangle$, if $a$ does not
occur free in $w$.

A term 
$\langle u \parallel_{a} v,\, w\rangle$
corresponds to a derivation 
\[ \infer{\langle u \parallel_{a} v,\, w\rangle : C \ET D } {
    \infer[^{1}]{ u \parallel_{a} v : C}{
\infer*{u:C}{\Gamma _{0}, [a:A \IMPL B]^{1}} & \infer*{v:C}{\Gamma _{1},
[a:B \IMPL A]^{1}}}  && \infer*{w:D}{\Gamma _{2}}}\] where
$\Gamma_{0}, \Gamma_{1}$ and $\Gamma_{2}$ 
are the open hypotheses.
Hence, the term
$\langle u,
w\rangle \parallel_{a} \langle v, w\rangle$
corresponds to
\[ \infer[^{1}]{
\langle u,
w\rangle \parallel_{a} \langle v, w\rangle : C \ET D} { \infer{\langle u,
w\rangle:C \ET D} {
\infer*{u:C}{\Gamma _{0}, [a:A \IMPL B]^{1}} & \infer*{w:D}{\Gamma
_{2}} } & \infer{\langle v, w\rangle: C \ET D} { \infer*{v:C}{\Gamma _{1}, [a:B
\IMPL A]^{1}} & \infer*{w:D}{\Gamma _{2}}} }
\]
 In which the formulas corresponding to the terms $a$ are discharged
above the appropriate premise of the $\COM$ rule applications. All open
hypotheses of the second derivation also occur in the first
derivation, and thus the free variables of the terms 
 $\langle u \parallel_{a} v,\, w\rangle$ and $ \langle u,
w\rangle \parallel_{a} \langle v, w\rangle$
are the same.

\item $\langle w, \,u \parallel_{a} v\rangle \mapsto \langle w,
u\rangle \parallel_{a} \langle w, v\rangle, \mbox{ if $a$ does not
occur free in $w$}$.

The case is symmetric to the previous one.


\item $(u\parallel_{a} v)\parallel_{b} w \mapsto (u\parallel_{b}
w)\parallel_{a} (v\parallel_{b} w),$ if the communication
complexity of $b$ is greater than $0$.

A term 
$(u\parallel_{a} v)\parallel_{b} w$
corresponds to a derivation
  \[ \vcenter{\infer[^{2}]{ (u \parallel_{a} v) \parallel_{b} w : E }
    {\deduce{u \parallel_{a} v : E}{{\cal D }_{1}} &&&
      \deduce{w:E}{{\cal P}_{2}}} } \qquad \text{where}\]
\[{\cal P}_{1}  \equiv \qquad \vcenter{\infer[^{1}]{ u \parallel_{a} v : E}{ \infer*{u:E}{\deduce{ [b:C \IMPL D]^{2}}{\Gamma _{0},
          [a:A \IMPL B]^{1},}} &&& \infer*{v:E} {\deduce{[b:C \IMPL D]^{2}}{\Gamma
          _{1}, [a:B \IMPL A]^{1},} }}} 
\]

\[{\cal P}_{2} \quad \equiv \qquad \quad  \vcenter{\infer*{w:E}{\Gamma _{2}, [b:D \IMPL C]^{2}}}\]

where
$\Gamma_{0}, \Gamma_{1}$ and $\Gamma_{2}$ are the open hypotheses.
Hence, the term
$ (u\parallel_{b}
w)\parallel_{a} (v\parallel_{b} w)$
corresponds to
  \[ \infer[^{3}]{ (u\parallel_{b} w)\parallel_{a} (v\parallel_{b} w)
      : E } {{\cal P}_{1} &&&& {\cal P}_{2}}\]
where ${\cal P}_{1} \equiv $
\[  \infer[^{1}]{ u\parallel_{b} w : E}{ \infer*{u:E}{\Gamma
          _{0}, [a:A \IMPL B]^{3}, [b:C \IMPL D]^{1}} &
        \infer*{w:E}{\Gamma _{2}, [b:D \IMPL C]^{1}}}  \] 
and ${\cal P}_{2} \equiv$ 
 \[   \infer[^{2}]{v\parallel_{b} w:E}{ \infer*{v:E}{\Gamma _{1}, [a:B \IMPL
          A]^{3}, [b:C \IMPL D]^{2}} & \infer*{w:E}{\Gamma _{2}, [b:D
          \IMPL C]^{2}} }\]
 In which the formulas corresponding to the terms $a$ and $b$ are discharged
above the appropriate premise of the $\COM$ rule applications. All open
hypotheses of the second derivation also occur in the first
derivation, and thus the free variables of
$(u\parallel_{a} v)\parallel_{b} w$
 and 
$ (u\parallel_{b}
w)\parallel_{a} (v\parallel_{b} w)$
are the same.

\item $w \parallel_{b} (u\parallel_{a} v) \mapsto (w\parallel_{b}
u)\parallel_{a} (w\parallel_{b} v)$, if the communication
complexity of $b$ is greater than $0$.

A term 
$w \parallel_{b} (u\parallel_{a} v)$
corresponds to a derivation 
  \[ \vcenter{\infer[^{2}]{ w \parallel_{b} (u\parallel_{a} v) : E }
    {\deduce{w:E}{{\cal P}_{1}} &&&&& \deduce{u \parallel_{a} v :
        E}{{\cal P}_{2}}}} \qquad \text{where}\]
\[ {\cal P}_{1} \quad \equiv \quad\vcenter{ \infer*{w:E}{\Gamma _{2},
      [b:D \IMPL C]^{2}}} \]
\[{\cal P}_{2} \equiv \qquad \vcenter{\infer[^{1}]{u \parallel_{a} v : E}{ \infer*{u:E}{\deduce{
[b:C \IMPL D]^{2}}{\Gamma _{0}, [a:A \IMPL B]^{1},}} &&& \infer*{v:E}{\deduce{[b:C
\IMPL D]^{2}}{\Gamma _{1}, [a:B \IMPL A]^{1},}}}}\]
Here we represent by
$\Gamma_{0}, \Gamma_{1}$ and $\Gamma_{2}$ the open hypotheses of the derivation.
Hence, the term  $(w\parallel_{b}
u)\parallel_{a} (w\parallel_{b} v)$
corresponds to
  \[ \vcenter{\infer[^{3}]{ (w\parallel_{b} u)\parallel_{a} (w\parallel_{b} v)
      : E } {{\cal P}_{1} && {\cal P}_{2}}} \qquad \text{where}\]
 
\[{\cal P}_{1} \; \equiv \;   \vcenter{\infer[^{1}]{ w\parallel_{b} u : E}{
        \infer*{w:E}{\Gamma _{2}, [b:D \IMPL C]^{1}} & \infer*{u:E}{\Gamma
          _{0}, [a:A \IMPL B]^{3}, [b:C \IMPL D]^{1}}}}  \] 
 \[{\cal P}_{2} \; \equiv \;   \vcenter{   \infer[^{2}]{w\parallel_{b} v:E}{ \infer*{w:E}{\Gamma _{2}, [b:D
          \IMPL C]^{2}} & \infer*{v:E}{\Gamma _{1}, [a:B \IMPL
          A]^{3}, [b:C \IMPL D]^{2}}} }\]
 In which the formulas corresponding to the terms $a$ and $b$ are discharged
above the appropriate premise of the $\COM$ rule application. The open
hypotheses of the second derivation also occur in the first
derivation, and thus the free variables of
the terms $w \parallel_{b} (u\parallel_{a} v)$ and $ (w\parallel_{b}
u)\parallel_{a} (w\parallel_{b} v)$
are the same.

\item $u\parallel_{a}v \mapsto u$, if $a$ does not occur in $u$.

A term 
$u\parallel_{a}v$
corresponds to a derivation 
\[\infer[^{1}]{ u\parallel_{a} v : C}{ \infer*{u:C}{\Gamma _{1}} &
\infer*{v:C}{\Gamma _{2}, [a:B \IMPL A]^{1}}} \] where $\Gamma _{1}$
and $\Gamma_{2}$ are all open hypotheses and no hypothesis $a: A \IMPL
B$ is discharged above $u: C$ by the considered $\COM$ application
because $a$ does not occur in $u$.  The term $u$ corresponds to the
 branch rooted at $u:C$, the open hypotheses of which are clearly
a subset of the open hypotheses of the whole derivation. Indeed the
$\COM$ application in the derivation corresponding to
$u\parallel_{a}v$ is redundant because $C$ can be derived from the
hypotheses $\Gamma _{1}$ alone. Thus, the free variables of $u$ are a
subset of the free variables of $u\parallel_{a}v$.


\item $u\parallel_{a}v \mapsto v$, if $a$ does not occur in $v$.

This case is symmetric to the previous one.

\item $\mathcal{C}[a\, u]\parallel_{a} \mathcal{D}[a\, v]\ \mapsto\
(\mathcal{D}[u^{b\langle
\sq{z}\rangle / \sq{y}}] \parallel_{a} \mathcal{C}[a\, u] )\, \parallel_{b}\, (\mathcal{C}[v^{b\langle
\sq{y}\rangle / \sq{z}}]\parallel_{a} \mathcal{D}[a\, v])$ where
$\mathcal{C}, \mathcal{D}$ are simple contexts; the displayed occurrences of $a$ are the rightmost both in $\mathcal{C}[a\,u]$ and in $\mathcal{D}[a\,v ]$ and
$b$ is fresh;
 $\sq{y}$ the sequence
of the free variables of $u$ which are bound in $\mathcal{C}[au]$;
$\sq{z}$ the sequence of the free variables of $v$ which are bound in
$\mathcal{D}[av]$; the communication complexity of $a$ is greater than $0$.

We introduce some notation in order to compactly represent the
handling of conjunctions of hypotheses corresponding to $\sq{y}$ and
$\sq{z}$.

\noindent \emph{Notation}.  For any finite set of formulas $\Theta =
\{\theta_{1} , \dots , \theta_{k}\}$ we define
  \[ \vcenter{ \infer=[\ICONG]{\SET\Theta}{\Theta} } \qquad \qquad \text{as} \qquad \qquad
\vcenter{ \infer{ \SET\Theta } { \infer* { } { \infer
{ } { \dots & \infer { \theta_{3} \ET \theta_{2} \ET
\theta_{1} } { \theta_{3} & \infer { \theta_{2} \ET \theta_{1}
} { \theta_{2} & \theta_{1} } } } } } }
  \] 
and $\vcenter{\infer=[\ECONG]{{\cal
      P}}{\SET\Theta}}$
as the derivation ${\cal P}$ in which all hypotheses contained in
$\Theta$ are derived as follows
  \[ \infer{\theta_{1}}{\SET\Theta} \qquad \cdots \qquad
\infer{\theta_{k}}{\infer*{\theta_{k-1} \ET
\theta_{k}}{\infer{\theta_{2} \ET \dots \ET
\theta_{k}}{\SET\Theta}}}
  \]

A term 
$\mathcal{C}[a\, u]\parallel_{a} \mathcal{D}[a\, v]$
corresponds to a derivation of the form 
\[
  \infer[^{1}]{\mathcal{C}[a\, u]\parallel_{a} \mathcal{D}[a\, v]: E}{
    \deduce{\mathcal{C}[a\, u] : E}{
      \deduce{
        {\cal P}_{3}
      }
      {
        \infer{
          au:B
        }
        {
           [a: A \IMPL B]^{1} & \deduce{
            u:A
          }
          {
            \deduce{
              {\cal P}_{1}
            } 
            {
              \sq{y} :\Gamma
            }
          }
        }
      }
    }
&&&&
\deduce{
  \mathcal{D}[a\, v]:E
}
{
  \deduce{
    {\cal P}_{4}
  }
  {
    \infer{
      av:A
    }
    {
      [a: B \IMPL A] ^{1} & \deduce{
        v:B
      }
      {
        \deduce{
          {\cal P}_{2}
        }
        {
          \sq{z} : \Delta
        }
      }
    }
  }
}
}
\] 
where $\Gamma$ and $\Delta$ are the open hypotheses of ${\cal P}_{1}$
discharged in ${\cal P}_{3}$, respectively, the open hypotheses of
${\cal P}_{2}$ discharged in ${\cal P}_{4}$; thus $\Gamma$ corresponds to
the sequence $\sq{y}$ that contains the free variables of $u$ which
are bound in $\mathcal{C}[au]$, while $\Delta$ corresponds to the
sequence $\sq{z}$ of all the free variables of $v$ which are bound in
$\mathcal{D}[av]$.  No hypotheses $a: A \IMPL B$ and $a:
B \IMPL A$ are discharged neither in ${\cal P}_{1}$ nor in ${\cal P}_{2}$ by the
considered $\COM$ application because $a$ does not occur neither in
$u$ nor in $v$.

The term $(\mathcal{D}[u^{b\langle
\sq{z}\rangle / \sq{y}}] \parallel_{a} \mathcal{C}[a\, u])\, \parallel_{b}\, (\mathcal{C}[v^{b\langle
\sq{y}\rangle / \sq{z}}]\parallel_{a} \mathcal{D}[a\, v])$ then corresponds
to the derivation
  \[ \infer[^{3}] {( \mathcal{D}[u^{b\langle
\sq{z}\rangle / \sq{y}}] \parallel_{a} \mathcal{C}[a\, u] )\, \parallel_{b}\, (\mathcal{C}[v^{b\langle
\sq{y}\rangle / \sq{z}}]\parallel_{a} \mathcal{D}[a\, v]) : E}
{\deduce{\mathcal{D}[u^{b\langle
\sq{z}\rangle / \sq{y}}] \parallel_{a} \mathcal{C}[a\, u] : E}{{\cal P}_{1}} &&& \deduce{\mathcal{C}[v^{b\langle \sq{y}\rangle /
        \sq{z}}]\parallel_{a} \mathcal{D}[a\, v] : E}{{\cal P}_{2}}}
  \]
where  ${\cal P}_{1} \equiv$
  \[\infer[^{1}] {\mathcal{C}[a\, u]\parallel_{a}
      \mathcal{D}[u^{b\langle \sq{z}\rangle / \sq{y}}] : E}
    {
      \deduce{\mathcal{D}[u^{b\langle \sq{z}\rangle / \sq{y}}] :
        E}{\deduce{{\cal P}_{4}}{\deduce{u^{b\langle \sq{z}\rangle /
              \sq{y}}:A}
          {\infer=[\ECONG]{{\cal P}_{1}^{b\langle \sq{z} \rangle , \sq{y}}}
            {\infer{b\langle \sq{z} \rangle:\SET\Gamma}{[b:
                  \SET\Delta \IMPL \SET\Gamma]^{3} &
                  \infer=[\ICONG]{\langle \sq{z}
                    \rangle:\SET\Delta}{\sq{z}: \Delta} }}}}}
&
\deduce{\mathcal{C}[a\, u] : E} { \deduce{ {\cal P}_{3} } {
          \infer{ au:B } { [a: A \IMPL B]^{1} & \deduce{ u:A } {
              \deduce{ {\cal P}_{1} } { \sq{y}:\Gamma } } } } } 
}
  \]
and ${\cal P}_{2} \equiv$
  \[ \infer[^{2}] {\mathcal{C}[v^{b\langle \sq{y}\rangle /
        \sq{z}}]\parallel_{a} \mathcal{D}[a\, v] : E} {
      \deduce{\mathcal{C}[v^{b\langle \sq{y}\rangle / \sq{z}}]:E}
      {\deduce{{\cal P}_{3}}{\deduce{v^{b\langle \sq{y}\rangle /
        \sq{z}}:B} {\infer=[\ECONG]{ {\cal P}_{2}^{b\langle \sq{y} \rangle , \sq{z}} } {\infer{b \langle \sq{y} \rangle : \SET\Delta}{ [b:
                  \SET\Gamma \IMPL \SET\Delta]^{3} &
                  \infer=[\ICONG]{\langle \sq{y} \rangle
                    :\SET\Gamma}{\sq{y}:\Gamma} } } } } } &
      \deduce{ \mathcal{D}[a\, v]:E } { \deduce{ {\cal P}_{4} } {
          \infer{av:A } { [a: B \IMPL A] ^{2} & \deduce{ v:B } {
              \deduce{ {\cal P}_{2} } { \sq{z}:\Delta } } } } }}
  \]



The resulting derivation is well defined and only needs the open
hypotheses of the derivation corresponding to the term
$\mathcal{C}[a\, u]\parallel_{a} \mathcal{D}[a\, v]$.  First, the
derivations corresponding to $\mathcal{C}[a\, u]$ and $\mathcal{D}[a\,
v]$ are the same as before the reduction. Consider then the
derivations corresponding to $\mathcal{D}[u^{b\langle \sq{z}\rangle /
\sq{y}}]$ and $\mathcal{C}[v^{b\langle \sq{y}\rangle / \sq{z}}]$. The
hypothesis needed in $ {\cal P}_{1}$ and ${\cal P}_{2}$ can be derived
by $(e\ET)$ from $\SET\Gamma$ and $\SET\Delta$, respectively. The
formulas $\SET\Gamma$ and $\SET\Delta$ can in turn be derived using
$b:\SET\Delta \IMPL \SET\Gamma $ and $b: \SET\Gamma \IMPL \SET\Delta$
along with the formulas $\SET\Delta$ and $\SET\Gamma$, in the
respective order. Thus it is possible to discharge the elements of
$\Gamma$ and $\Delta$ in the derivations ${\cal P}_{3}$ and ${\cal
P}_{4}$, respectively, exactly as in the redex derivation
corresponding to $\mathcal{C}[a\, u]\parallel_{a} \mathcal{D}[a\, v]$.
\end{enumerate}
\end{proof}



\noindent \textbf{Proposition~\ref{proposition-boundhyp}}~(Bound Hypothesis Property)\textbf{.}
Suppose
$$x_{1}^{A_{1}}, \ldots, x_{n}^{A_{n}}\vdash t: A$$
$t\in\nf$ is a simply typed $\lambda$-term and  $z: B$ a variable occurring bound in $t$. Then one of the following holds:
\begin{enumerate}
\item $B$ is a proper subformula of a prime factor of $A$.
\item $B$ is a strong subformula of one among $A_{1},\ldots, A_{n}$.
\end{enumerate}
\begin{proof}
By induction on $t$.  
\begin{itemize} 
\item $t=x_{i}^{A_{i}}$, with $1\leq i\leq n$. Since by hypothesis $z$ must occur bound in $t$,  this case is impossible.

\item $t=\lambda x^{T} u$, with $A=T\rightarrow U$. If $z=x^{T}$,
since $A$ is a prime factor of itself, we are done. If $z\neq x^{T}$,
then $z$ occurs bound in $u$ and by induction hypothesis applied to
$u: U$, we have two possibilities: i) $B$ is a proper subformula of
a prime factor of $U$ and thus a proper subformula of a prime factor
-- $A$ itself -- of $A$; ii) $B$ already satisfies 2., and we are
done, or $B$ is a strong subformula of $T$, and thus it satisfies 1.

\item $t=\langle u_{1}, u_{2}\rangle$, with $A=T_{1}\land T_{2}$. Then
$z$ occurs bound in $u$ or $v$ and, by induction hypothesis applied to
$u_{1}: T_{1}$ and $u_{2}: T_{2}$, we have two possibilities: i) $B$
is a proper subformula of a prime factor of $T_{1}$ or $T_{2}$, and
thus $B$ is a proper subformula of a prime factor of $A$ as well; ii)
$B$ satisfies 2. and we are done.

\item $t=\efq{P}{u}$, with $A=P$. Then $z$ occurs bound in $u$. Since
$\bot$ has no proper subformula, by induction hypothesis applied to
$u: \bot$, we have that $B$ satisfies 2. and we are done.

\item $t=x_{i}^{A_{i}}\, \xi_{1}\ldots \xi_{m}$, where $m>0$ and each
$\xi_{j}$ is either a term or a projection $\pi_{k}$. Since $z$ occurs
bound in $t$, it occurs bound in some term $\xi_{j}: T$, where $T$ is
a proper subformula of $A_{i}$. By induction hypothesis applied to
$\xi_{j}$, we have two possibilities: i) $B$ is a proper subformula of
a prime factor of $T$ and, by Definition \ref{definition-strongsubf},
$B$ is a strong subformula of $A_{i}$.
ii) $B$ satisfies 2. and we are done.
\end{itemize}
\end{proof}

\noindent \textbf{Proposition~\ref{proposition-app}}\textbf{.}
  Suppose that $t\in \nf$ is a simply typed $\lambda$-term and
$$x_{1}^{A_{1}}, \ldots, x_{n}^{A_{n}}, z^{B}\vdash t: A$$
Then one of the following holds:
\begin{enumerate}
\item \emph{Every occurrence of $z^{B}$ in $t$ is of the form
    $z^{B}\, \xi$ for some proof term or projection $\xi$.}
\item \emph{$B=\bot$ or $B$ is a subformula of $A$ or a proper
    subformula of one among the formulas $A_{1}, \ldots, A_{n}$.}
\end{enumerate}

\begin{proof}
By induction on $t$. 
\begin{itemize} 
\item $t=x_{i}^{A_{i}}$, with $1\leq i\leq n$. Trivial.
\item $t=z^{B}$. This means that $B=A$, and we are done.
\item $t=\lambda x^{T} u$, with $A=T\rightarrow U$. By induction hypothesis applied to $u: U$, we have two possibilities:  i)  every occurrence of $z^{B}$ in $u$ is of the form $z^{B}\, \xi$, and we are done; ii) $B=\bot$ or $B$ is a  subformula of  $U$,  and hence of $A$, or a  proper subformula of one among the formulas $A_{1}, \ldots, A_{n}$, and we are done again.
\item $t=\langle u_{1}, u_{2}\rangle$, with $A=T_{1}\land T_{2}$. By
induction hypothesis applied to $u_{1}: T_{1}$ and $u_{2}: T_{2}$, we
have two possibilities: i) every occurrence of $z^{B}$ in $u_{1}$ and
$u_{2}$ is of the form $z^{B}\, \xi$, and we are done; ii) $B=\bot$ or
$B$ is a subformula of $T_{1}$ or $T_{2}$, and hence of $A$, or a
proper subformula of one among the formulas $A_{1}, \ldots, A_{n}$,
and we are done again.
\item $t=\efq{P}{u}$, with $A=P$.  By induction hypothesis applied to $u: \bot$, we have two possibilities: i)  every occurrence of $z^{B}$ in $u$ is of the form $z^{B}\, \xi$, and we are done; ii) $B=\bot$  or a  proper subformula of one among 
  $A_{1}, \ldots, A_{n}$, and we are done again.
\item $t=x_{i}^{A_{i}}\, \xi_{1}\ldots \xi_{m}$, where $m>0$ and each
$\xi_{j}$ is either a term or a projection $\pi_{k}$. Suppose there is
an $i$ such that in the term $\xi_{j}: T_{j}$ not every occurrence of
$z^{B}$ in $u$ is of the form $z^{B}\, \xi$. If $B=\bot$, we are
done. If not, then by induction hypothesis $B$ is a subformula of
$T_{j}$ or a proper subformula of one among $A_{1}, \ldots,
A_{n}$. Since $T_{j}$ is a proper subformula of $A_{i}$, in both cases
$B$ is a proper subformula of one among $A_{1}, \ldots, A_{n}$.
\item $t=z^{B}\, \xi_{1}\ldots \xi_{m}$, where $m>0$ and each
$\xi_{i}$ is either a term or a projection $\pi_{j}$. Suppose there is
an $i$ such that in the term $\xi_{i}: T_{i}$ not every occurrence of
$z^{B}$ in $u$ is of the form $z^{B}\, \xi$. If $B=\bot$, we are
done. If not, then by induction hypothesis $B$ is a subformula of
$T_{i}$ or a proper subformula of one among $A_{1}, \ldots,
A_{n}$. But the former case is not possible, since $T_{i}$ is a proper
subformula of $B$, hence the latter holds.
\end{itemize}
\end{proof}



\noindent \textbf{Proposition~\ref{proposition-parallelform}}~(Parallel Normal Form Property)\textbf{.}

  Suppose $t\in \nf$, then it is parallel form.
  \begin{enumerate}
  \item \emph{Every occurrence of $z^{B}$ in $t$ is of the form
      $z^{B}\, \xi$ for some proof term or projection $\xi$.}
  \item \emph{$B=\bot$ or $B$ is a subformula of $A$ or a proper
      subformula of one among the formulas $A_{1}, \ldots, A_{n}$.}
  \end{enumerate}
\begin{proof}
By induction on $t$. 
\begin{itemize}

\item $t$ is a variable $x$. Trivial. 

\item $t=\lambda x\, u$. Since $t$ is normal, $u$ cannot be of the form $v\parallel_{a} w$, otherwise one could apply the permutation
$$t=\lambda x\,(\Ecrom{a}{v}{w})  \mapsto \Ecrom{a}{\lambda x\,v}{\lambda x\, w} $$
and $t$ would not be in normal form. Hence, by induction hypothesis $u$ must be a simply typed $\lambda$-term, QED. \\  

\item $t=\langle u_{1}, u_{2}\rangle$. Since $t$ is normal, neither $u_{1}$ nor $u_{2}$  can be of the form $v\parallel_{a} w$, otherwise one could apply one of  the permutations
$$\langle v \parallel_{a} w,\, u_{2}\rangle \mapsto \langle v, u_{2}\rangle \parallel_{a} \langle w, u_{2}\rangle$$
$$\langle u_{1}, \,v \parallel_{a} w\rangle \mapsto \langle u_{1}, v\rangle \parallel_{a} \langle u_{1}, w\rangle$$
and $t$ would not be in normal form. Hence, by induction hypothesis $u_{1}$ and $u_{2}$ must 
be simply typed $\lambda$-terms, QED.

\item $t=u\, v$. Since $t$ is normal, neither $u$ nor $v$  can be of the form $w_{1}\parallel_{a} w_{2}$, otherwise one could apply one of  the permutations
$$(\Ecrom{a}{w_{1}}{w_{2}})\, v \mapsto \Ecrom{a}{w_{1}\, v}{w_{2}\, v}$$
$$u\, (\Ecrom{a}{w_{1}}{w_{2}})  \mapsto \Ecrom{a}{u\, w_{1}}{u\, w_{2}}$$
and $t$ would not be in normal form. Hence, by induction hypothesis $u_{1}$ and $u_{2}$ must 
be simply typed $\lambda$-terms, QED.

\item $t=\efq{P}{u}$. Since $t$ is normal, $u$  cannot be of the form $w_{1}\parallel_{a} w_{2}$, otherwise one could apply  the permutation
$$\efq{P}{\Ecrom{a}{w_{1}}{w_{2}}}  \mapsto \Ecrom{a}{\efq{P}{w_{1}}}{\efq{P}{w_{2}}}$$
and $t$ would not be in normal form. Hence, by induction hypothesis $u_{1}$ and $u_{2}$ must 
be simply typed $\lambda$-terms, QED. 

\item $t=u\, \pi_{i}$. Since $t$ is normal,  $u$  cannot be of the form $v\parallel_{a} w$, otherwise one could apply  the permutation
$$(\Ecrom{a}{v}{w})\pi_{i}  \mapsto \Ecrom{a}{v\pi_{i}}{w\pi_{i}} $$
and $t$ would not be in normal form. Hence, by induction hypothesis $u$ must be a simply typed $\lambda$-term, which is the thesis.
\item $t=u\parallel_{a} v$. By induction hypothesis the thesis holds for $u$ and $v$ and hence trivially for $t$.  
\end{itemize}
\end{proof}



\noindent \textbf{Theorem~\ref{theorem-subformula}}~(Subformula Property)\textbf{.}
Suppose
$$x_{1}^{A_{1}}, \ldots, x_{n}^{A_{n}}\vdash t: A$$
and $t\in \nf$. 
Then: 
\begin{enumerate}
\item
For each communication variable $a$ occurring bound in  $t$ and with communication kind $B, C$, the prime factors of $B$ and $C$ are proper subformulas of the formulas $A_{1}, \ldots, A_{n}, A$. 
\item The type of any subterm of $t$ which is not a bound communication variable is either a subformula or a conjunction of subformulas of the formulas $A_{1}, \ldots, A_{n}, A$. 
\end{enumerate}
\begin{proof} By Proposition \ref{proposition-parallelform}
$t = t_{1}\parallel_{a_{1}} t_{2}\parallel_{a_{2}}\ldots \parallel_{a_{n}} t_{n+1}$
and each term $t_{i}$, for $1\leq i\leq n+1$, is a simply typed $\lambda$-term.
By induction on $t$. 
\begin{itemize} 
\item $t=x_{i}^{A_{i}}$, with $1\leq i\leq n$. Trivial. 

\item $t=\lambda x^{T} u$, with $A=T\rightarrow U$. By Proposition \ref{proposition-parallelform}, $t$ is a simply typed $\lambda$-term, so $t$ contains no bound communication variable. Moreover, by induction hypothesis applied to $u: U$, the type of any subterm of $u$ which is not a bound communication variable is either a subformula or a conjunction of subformulas of the formulas $T, A_{1}, \ldots, A_{n},  U$ and therefore of the formulas $A_{1}, \ldots, A_{n},  A$.


\item $t=\langle u_{1}, u_{2}\rangle$, with $A=T_{1}\land T_{2}$. By Proposition \ref{proposition-parallelform}, $t$ is a simply typed $\lambda$-term, so $t$ contains no bound communication variable. Moreover, by induction hypothesis applied to $u_{1}: T_{1}$ and $u_{2}: T_{2}$, the type of any subterm of $u$ which is not a bound communication variable is either a subformula or a conjunction of subformulas of the formulas $ A_{1}, \ldots, A_{n}, T_{1}, T_{2}$ and hence of the formulas $A_{1}, \ldots, A_{n},  A$.

\item $t= u_{1}\parallel_{b} u_{2}$.  Let $C, D$ be the communication
  kind of $b$, we first show that the communication
    complexity of $b$ is 
  $0$.
 We reason by contradiction and assume that it is greater than $0$. Since $t\in\nf$ by hypothesis, it follows
 from Prop.~\ref{proposition-parallelform} that $u_{1}$ and $u_{2}$ are either
simply typed $\lambda$-terms or of the form $v\parallel_{c} w$. The
second case is not possible, otherwise a permutation reduction could
be applied to $t\in \nf$. Thus $u_{1}$ and $u_{2}$ are simply typed
$\lambda$-terms. Since the communication complexity of $b$ is greater
than $0$, the types $C\rightarrow D$ and $D\rightarrow C$ are not
subformulas of $A_{1}, \ldots, A_{n}, A$. By Prop.~\ref{proposition-app}, every occurrence of $b^{C\rightarrow D}$ in
$u_{1}$ is of the form $b^{C\rightarrow D} v$ and every occurrence of
$b^{D\rightarrow C}$ in $u_{2}$ is of the form $b^{D\rightarrow C} w$.
Hence, we can write
$$u_{1}=\mathcal{C}[b^{C\rightarrow D} v] \qquad u_{2}=\mathcal{D}[b^{D\rightarrow C} w]$$
where $\mathcal{C}, \mathcal{D}$ are simple contexts and $b$ is rightmost. Hence a cross reduction of $t$ can be performed,
which contradicts the fact that $t\in\nf$.  Since we have established
that the communication complexity of $b$ is $0$, the prime factors of
$C$ and $D$ must be proper subformulas of $A_{1}, \ldots, A_{n},
A$. Now, by induction hypothesis applied to $u_{1}: A$ and $u_{2}: A$,
for each communication variable $a^{F\rightarrow G}$ occurring bound
in $t$, the prime factors of $F$ and $G$ are proper subformulas of the
formulas $ A_{1},  \ldots, A_{n}, A, C\rightarrow D, D\rightarrow C$
and thus of the formulas $A_{1}, \ldots, A_{n}, A$; moreover, the
type of any subterm of $u_{1}$ or $u_{2}$ which is not a communication
variable is either a subformula or a conjunction of subformulas of the
formulas $ A_{1}, \ldots, A_{n}, C\rightarrow D, D\rightarrow C$ and
thus of the formulas $A_{1},  \ldots, A_{n}, A$.

\item $t=x_{i}^{A_{i}}\, \xi_{1}\ldots \xi_{m}$, where $m>0$ and each
$\xi_{j}: T_{j}$ is either a term or a projection $\pi_{k}$ and
$T_{j}$ is a subformula of $A_{i}$. By Proposition
\ref{proposition-parallelform}, $t$ is a simply typed $\lambda$-term,
so $t$ contains no bound communication variable. By induction
hypothesis applied to each $\xi_{j}: T_{j}$, the type of any subterm
of $t$ which is not a bound communication variable is either a
subformula or a conjunction of subformulas of the formulas $ A_{1},
\ldots, A_{n}, T_{1}, \ldots, T_{m}$ and thus of the formulas
$A_{1},  \ldots, A_{n}, A$.

\item $t=\efq{P}{u}$, with $A=P$. Then $u=x_{i}^{A_{i}}\, \xi_{1}\ldots \xi_{m}$, where $m>0$ and each $\xi_{j}$ is either a term or a projection $\pi_{k}$. Hence, $\bot$ is a subformula of $A_{i}$. Finally, by the proof of the previous case, we obtain the thesis for $t$.
\end{itemize}
\end{proof}



\noindent \textbf{Proposition~\ref{proposition-normpar}}\textbf{.}
Let $t: A$ be any term. Then $t\mapsto^{*} t'$, where $t'$ is a parallel form. 
\begin{proof}
By induction on $t$. 
\begin{itemize}
\item  $t$ is a variable $x$. Trivial. 
\item $t=\lambda x\, u$. By induction hypothesis, 
$$u\mapsto^{*} u_{1}\parallel_{a_{1}} u_{2}\parallel_{a_{2}}\ldots \parallel_{a_{n}} u_{n+1}$$
and each term $u_{i}$, for $1\leq i\leq n+1$, is a simply typed $\lambda$-term. Applying $n$ times the permutations
we obtain
$$t\mapsto^{*} \lambda x\, u_{1}\parallel_{a_{1}}\lambda x\, u_{2}\parallel_{a_{2}}\ldots \parallel_{a_{n}}\lambda x\, u_{n+1}$$
which is the thesis.
\item $t=u\, v$. By induction hypothesis, 
$$u\mapsto^{*} u_{1}\parallel_{a_{1}} u_{2}\parallel_{a_{2}}\ldots \parallel_{a_{n}} u_{n+1}$$
$$v\mapsto^{*} v_{1}\parallel_{b_{1}} v_{2}\parallel_{b_{2}}\ldots \parallel_{b_{m}} v_{m+1}$$
and each term $u_{i}$ and $v_{i}$, for $1\leq i\leq n+1, m+1$, is a simply typed $\lambda$-term. Applying $n+m$ times the permutations
we obtain
\[
\begin{aligned}
t &\mapsto^{*} (u_{1}\parallel_{a_{1}} u_{2}\parallel_{a_{2}}\ldots \parallel_{a_{n}} u_{n+1})\, v \\
&\mapsto^{*}  u_{1}\, v \parallel_{a_{1}} u_{2}\, v \parallel_{a_{2}}\ldots \parallel_{a_{n}} u_{n+1}\, v\\
&\mapsto^{*} u_{1}\, v_{1} \parallel_{b_{1}}  u_{1}\, v_{2}\parallel_{b_{2}}\ldots \parallel_{b_{m}} u_{1}\, v_{m+1} \parallel_{a_{1}} \ldots
\\
& \qquad  \, \ldots \parallel_{a_{n}}  u_{n+1}\, v_{1} \parallel_{b_{1}}  u_{n+1}\, v_{2} \parallel_{b_{2}}\ldots
 \parallel_{b_{m}}  u_{n+1}\, v_{m+1}
\end{aligned}
\]

\item $t=\langle u, v\rangle$. By induction hypothesis, 
$$u\mapsto^{*} u_{1}\parallel_{a_{1}} u_{2}\parallel_{a_{2}}\ldots \parallel_{a_{n}} u_{n+1}$$
$$v\mapsto^{*} v_{1}\parallel_{b_{1}} v_{2}\parallel_{b_{2}}\ldots \parallel_{b_{m}} v_{m+1}$$
and each term $u_{i}$ and $v_{i}$, for $1\leq i\leq n+1, m+1$, is a simply typed $\lambda$-term. Applying $n+m$ times the permutations
we obtain
  \[
    \begin{aligned}
      t &\mapsto^{*}\langle u_{1}\parallel_{a_{1}} u_{2}\parallel_{a_{2}}\ldots \parallel_{a_{n}} u_{n+1},\, v \rangle\\
      &\mapsto^{*} \langle u_{1}, v\rangle \parallel_{a_{1}} \langle u_{2}, v \rangle\parallel_{a_{2}}\ldots \parallel_{a_{n}} \langle u_{n+1}, v\rangle\\
      &\mapsto^{*} \langle u_{1}, v_{1}\rangle \parallel_{b_{1}}
      \langle u_{1}, v_{2}
      \rangle\parallel_{b_{2}}\ldots \parallel_{b_{m}} \langle u_{1},
      v_{m+1}\rangle \parallel_{a_{1}} \ldots
      \\
      & \qquad \, \ldots
      \parallel_{a_{n}} \langle u_{n+1},
      v_{1}\rangle \parallel_{b_{1}} \langle u_{n+1}, v_{2}
      \rangle\parallel_{b_{2}}\ldots
      \\
      & \qquad \, \ldots \parallel_{b_{m}} \langle u_{n+1},
      v_{m+1}\rangle
    \end{aligned}
  \]

\item $t=u\, \pi_{i}$. By induction hypothesis, 
$$u\mapsto^{*} u_{1}\parallel_{a_{1}} u_{2}\parallel_{a_{2}}\ldots \parallel_{a_{n}} u_{n+1}$$
and each term $u_{i}$, for $1\leq i\leq n+1$, is a simply typed $\lambda$-term. Applying $n$ times the permutations
we obtain
$$t\mapsto^{*}  u_{1}\, \pi_{i}\parallel_{a_{1}} u_{2}\,
\pi_{i}\parallel_{a_{2}}\ldots \parallel_{a_{n}} u_{n+1} \, \pi_{i}.$$

\item $t=\efq{P}{u}$. By induction hypothesis, 
$$u\mapsto^{*} u_{1}\parallel_{a_{1}} u_{2}\parallel_{a_{2}}\ldots \parallel_{a_{n}} u_{n+1}$$
and each term $u_{i}$, for $1\leq i\leq n+1$, is a simply typed $\lambda$-term. Applying $n$ times the permutations 
we obtain
$$t\mapsto^{*}  \efq{P}{u_{1}}\parallel_{a_{1}} \efq{P}{u_{2}}\parallel_{a_{2}}\ldots \parallel_{a_{n}} \efq{P}{u_{n+1}}$$
QED
\end{itemize}
\end{proof}


\end{document}